\documentclass[11pt]{amsart}

\usepackage[bookmarks=false]{hyperref}
\usepackage{float}
\usepackage[T1]{fontenc}
\usepackage[utf8]{inputenc}
\usepackage{amsmath,amssymb,stmaryrd,xspace,pdfsync,mathrsfs,graphicx,refcount,verbatim}
\usepackage{microtype}
\usepackage{tikz}
\usepackage[mathscr]{euscript}
\usetikzlibrary{decorations.text,arrows,snakes,shapes,fit,shadows,calc,patterns,intersections}
\usepackage{makecell}

\usepackage{fullpage}

\usepackage{tikz}
\usepackage{enumitem}
\usepackage{marvosym}

\usetikzlibrary{decorations.text,arrows,snakes,shapes,fit,shadows,calc,patterns,intersections}

\setlist[1]{itemsep=.5ex,leftmargin=*,topsep=0.5ex,itemsep=.3ex}
\setlist[1,itemize]{label=--}

\newcommand{\efgame}{Ehrenfeucht-Fra\"iss\'e\xspace}
\newcommand\nat{\ensuremath{\mathbb{N}}\xspace}

\newcommand\Cs{\ensuremath{\mathcal{C}}\xspace}

\newcommand\Cslev[1]{\ensuremath{\Cs_{#1}}\xspace}

\newcommand\Ss{\ensuremath{\mathcal{S}}\xspace}
\newcommand\Fs{\ensuremath{\mathcal{F}}\xspace}

\newcommand\Ts{\ensuremath{\mathcal{T}}\xspace}

\newcommand\Ms{\ensuremath{\mathcal{M}}\xspace}

\newcommand\Sat{\ensuremath{\mathord{\mathrm{Calc}}}}

\newcommand{\sic}[1]{\ensuremath{\Sigma_{#1}}\xspace}

\newcommand{\pic}[1]{\ensuremath{\Pi_{#1}}\xspace}

\newcommand{\bsc}[1]{\ensuremath{\mathcal{B}\Sigma_{#1}}\xspace}

\newcommand{\sicu}{\ensuremath{\Sigma_{1}}\xspace}

\newcommand{\picu}{\ensuremath{\Pi_{1}}\xspace}

\newcommand{\bscu}{\ensuremath{\mathcal{B}\Sigma_{1}}\xspace}

\newcommand{\sicd}{\ensuremath{\Sigma_{2}}\xspace}

\newcommand{\siwsd}{\ensuremath{\Sigma_{2}(<,+1)}\xspace}

\newcommand{\picd}{\ensuremath{\Pi_{2}}\xspace}

\newcommand{\bscd}{\ensuremath{\mathcal{B}\Sigma_{2}}\xspace}

\newcommand{\bswsd}{\ensuremath{\mathcal{B}\Sigma_{2}(<,+1)}\xspace}

\newcommand{\sict}{\ensuremath{\Sigma_{3}}\xspace}

\newcommand{\siwst}{\ensuremath{\Sigma_{3}(<,+1)}\xspace}

\newcommand{\pict}{\ensuremath{\Pi_{3}}\xspace}

\newcommand{\bsct}{\ensuremath{\mathcal{B}\Sigma_{3}}\xspace}

\newcommand{\sici}{\ensuremath{\Sigma_{\lowercase{i}}}\xspace}

\newcommand{\pici}{\ensuremath{\Pi_{\lowercase{i}}}\xspace}

\newcommand{\bsci}{\ensuremath{\mathcal{B}\Sigma_{\lowercase{i}}}\xspace}

\newcommand{\fo}{\ensuremath{\textup{FO}}\xspace}

\newcommand\sieq[2]{\ensuremath{\lesssim^{#1}_{#2}}\xspace}
\newcommand\ksieq[1]{\sieq{k}{#1}}

\newcommand\bceq[2]{\ensuremath{\cong^{#1}_{#2}}\xspace}
\newcommand\kbceq[1]{\bceq{k}{#1}}

\let\leq\leqslant

\let\geq\geqslant

\newcommand\content[1]{\ensuremath{\contentmorphism(#1)}}

\newcommand\contentmorphism{\ensuremath{\textsf{alph}}}

\newcommand\iword{\ensuremath{\omega}-word\xspace}

\newcommand\ilang{\ensuremath{\omega}-language\xspace}
\newcommand\iLang{\ensuremath{\omega}-Language\xspace}
\newcommand\iwords{\ensuremath{\omega}-words\xspace}
\newcommand\iWords{\ensuremath{\omega}-Words\xspace}
\newcommand\ilangs{\ensuremath{\omega}-languages\xspace}
\newcommand\iLangs{\ensuremath{\omega}-Languages\xspace}

\newcommand\iSemi{\ensuremath{\omega}-Semigroup\xspace}
\newcommand\iSemis{\ensuremath{\omega}-Semigroups\xspace}
\newcommand\isemi{\ensuremath{\omega}-semigroup\xspace}
\newcommand\isemis{\ensuremath{\omega}-semigroups\xspace}

\newcommand\chain{chain\xspace}
\newcommand\chains{chains\xspace}

\newcommand\qchain[1]{\ensuremath{\sic{#1}}-chain\xspace}
\newcommand\qchains[1]{\ensuremath{\sic{#1}}-chains\xspace}

\newcommand\Chains{Chains\xspace}
\newcommand\qChains[1]{\ensuremath{\sic{#1}}-Chains\xspace}

\newcommand\ichain{\qchain{\lowercase{i}}}

\newcommand\dchain{\qchain{2}}

\newcommand\ichains{\qchains{\lowercase{i}}}

\newcommand\dchains{\qchains{2}}
\newcommand\tchains{\qchains{3}}

\newcommand\iChains{\qChains{\lowercase{i}}}

\newcommand\qjuns[1]{\ensuremath{\sic{#1}}-junctures\xspace}

\newcommand\djuns{\qjuns{2}}

\renewcommand\bar\overline

\tikzstyle{nor}=[minimum size=0.35cm,draw,rounded rectangle,inner sep=2pt]
\tikzstyle{nod}=[minimum size=0.35cm,draw,circle,inner sep=2pt]
\tikzstyle{nok}=[minimum size=0.45cm,draw,circle,inner sep=0pt]
\tikzstyle{nof}=[minimum size=0.35cm,draw,circle,double,double
distance=1pt]
\tikzstyle{port}=[minimum size=0.35cm,draw,thick,rectangle,inner sep=2pt]
\tikzstyle{nop}=[minimum size=0.35cm,draw,thick,rectangle,inner sep=1pt,rotate=90]

\tikzstyle{nol}=[minimum size=0.35cm,draw,rounded rectangle,inner sep=1pt,rotate=90]

\tikzstyle{ar}=[line width=0.5pt,->,double]
\tikzstyle{siar}=[line width=1.5pt,->]
\tikzstyle{sars}=[line width=1pt,->,double]

\newtheorem{theorem}{Theorem}[section]
\newtheorem{lemma}[theorem]{Lemma}
\newtheorem{corollary}[theorem]{Corollary}

\newtheorem{proposition}[theorem]{Proposition}
\newtheorem{fact}[theorem]{Fact}
\newtheorem{remark}[theorem]{Remark}

\title{Quantifier Alternation for Infinite Words}

\author{Théo~Pierron, Thomas~Place and Marc~Zeitoun}
\address{Univ. Bordeaux, LaBRI, UMR 5800, F-33400 Talence, France}
\begin{document}

\begin{abstract}
  We investigate the expressive power of quantifier alternation hierarchy of
  first-order logic over words. This hierarchy includes the classes \sici
  (sentences having at most $i$ blocks of quantifiers starting with an
  $\exists$) and \bsci (Boolean combinations of \sici sentences). So far, this
  expressive power has been effectively characterized  for the lower levels only.
  Recently, a breakthrough was made over finite words, and decidable
  characterizations were obtained for \bscd and \sict, by relying on a
  decision problem called separation, and solving it for \sicd.

  The contribution of this paper is a generalization of these results to the
  setting of infinite words: we solve separation for \sicd  and \sict, and
  obtain decidable characterizations of \bscd and \sict as consequences.
\end{abstract}

\maketitle

\section{Introduction}
\label{sec:intro}
Regular word languages form a robust class, as they can be defined either by
operational, algebraic, or logical means: they are exactly those that can be
defined equivalently by finite state machines (operational view), morphisms
into finite algebras (algebraic view) and monadic second order (MSO)
sentences~\cite{BuchiMSO,TrakhMSO,ElgotMSO,BuchiMSOInf} (logical view). To
understand in depth the structure of this class, it is natural to classify
its languages according to their descriptive complexity. The problem is to
determine how complicated has to be a sentence to describe a given input
language. This is a decision problem parametrized by a fragment
of MSO: given an input language, can it be expressed in the fragment? This
problem is called \emph{membership} (is the language a \emph{member}
of the class defined by the fragment?).

The seminal result in this field is the membership algorithm for
first-order logic (\fo) over finite words, which is arguably the
most prominent fragment of~MSO. This algorithm was obtained in two
steps. McNaughton and Papert~\cite{mnpfo} observed that the languages
definable in \fo are exactly the \emph{star-free languages}: those
that may be expressed by a regular expression in which complement is
allowed while the Kleene star is disallowed. Furthermore, an earlier
result of Schützenberger~\cite{sfo} shows that star-free languages
are exactly the ones whose syntactic monoid is aperiodic. The syntactic
monoid is a finite algebra that can be computed from any input regular
language, and aperiodicity can be formulated as an equation that has to
be satisfied by all elements of this algebra. Therefore,
Schützenberger's result makes it possible to decide whether a regular
language is star-free (and therefore definable in \fo by
McNaughton-Papert's result).

Following this first result, the attention turned to a deeper
question: given an \fo-definable language, find the ``simplest''
\fo-sentences that define it. The standard complexity measure for \fo
sentences is their quantifier alternation, which counts the number of
switches between blocks of $\exists$ and $\forall$~quantifiers. This
measure is justified not only because it is intuitively difficult to
understand a sentence with many alternations, but also because the
nonelementary complexity of standard problems for
\fo~\cite{StockmeyerMeyer} (\emph{e.g}, satisfiability) is tied to
quantifier alternation. In~summary, we classify \fo definable
languages by counting the number of quantifier alternations needed to
define them and we want to be able to decide the level of a given
language (which amounts to solving membership for each~level).

This leads to define the following fragments of \fo: an \fo sentence is \sici
if its prenex normal form has at most $i$ blocks of $\exists$ or $\forall$ quantifiers and
starts with a block of existential ones. Note that \sici is not closed
under complement (the negation of a \sici sentence is called a \pici sentence).
A sentence is \bsci if it is a Boolean combination of \sici sentences. Clearly,
we have $\sici\subseteq\bsci\subseteq\sic{i+1}$, and these inclusions are
known to be strict~\cite{BroKnaStrict,ThomStrict}:
$\sici\subsetneq\bsci\subsetneq\sic{i+1}$. See the figure.

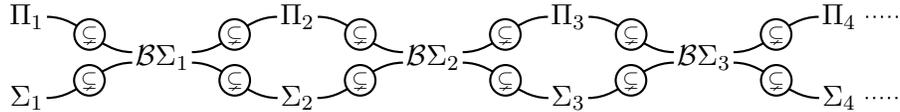
\tikzstyle{non}=[inner sep=1pt]
\tikzstyle{tag}=[draw,fill=white,circle,inner sep=1pt]
\begin{figure}[H]
  \centering
  \begin{tikzpicture}[scale=.9]
    \node[non] (s1) at (0.0,-0.6) {\sicu};
    \node[non] (p1) at (0.0,0.6) {\picu};
    \node[non] (b1) at (2.0,0.0) {\bscu};

    \node[non] (s2) at ($(b1)+(2.0,-0.6)$) {\sicd};
    \node[non] (p2) at ($(b1)+(2.0,0.6)$) {\picd};
    \node[non] (b2) at ($(b1)+(4.0,0.0)$) {\bscd};

    \node[non] (s3) at ($(b2)+(2.0,-0.6)$) {\sict};
    \node[non] (p3) at ($(b2)+(2.0,0.6)$) {\pict};
    \node[non] (b3) at ($(b2)+(4.0,0.0)$) {\bsct};

    \node[non] (s4) at ($(b3)+(2.0,-0.6)$) {\sic{4}};
    \node[non] (p4) at ($(b3)+(2.0,0.6)$) {\pic{4}};

    \draw[thick] (s1.east) to [out=0,in=180] node[tag] {\scriptsize
      $\subsetneq$} (b1.190);
    \draw[thick] (p1.east) to [out=0,in=180] node[tag] {\scriptsize
      $\subsetneq$} (b1.170);

    \draw[thick] (b1.-10) to [out=0,in=180] node[tag] {\scriptsize
      $\subsetneq$} (s2.west);
    \draw[thick] (b1.10) to [out=0,in=180] node[tag] {\scriptsize
      $\subsetneq$} (p2.west);
    \draw[thick] (s2.east) to [out=0,in=180] node[tag] {\scriptsize
      $\subsetneq$} (b2.190);
    \draw[thick] (p2.east) to [out=0,in=180] node[tag] {\scriptsize
      $\subsetneq$} (b2.170);

    \draw[thick] (b2.-10) to [out=0,in=180] node[tag] {\scriptsize
      $\subsetneq$} (s3.west);
    \draw[thick] (b2.10) to [out=0,in=180] node[tag] {\scriptsize
      $\subsetneq$} (p3.west);
    \draw[thick] (s3.east) to [out=0,in=180] node[tag] {\scriptsize
      $\subsetneq$} (b3.190);
    \draw[thick] (p3.east) to [out=0,in=180] node[tag] {\scriptsize
      $\subsetneq$} (b3.170);

    \draw[thick] (b3.-10) to [out=0,in=180] node[tag] {\scriptsize
      $\subsetneq$} (s4.west);
    \draw[thick] (b3.10) to [out=0,in=-180] node[tag] {\scriptsize
      $\subsetneq$} (p4.west);

    \draw[thick,dotted] ($(s4.east)+(0.1,0.0)$) to
    ($(s4.east)+(0.6,0.0)$);
    \draw[thick,dotted] ($(p4.east)+(0.1,0.0)$) to
    ($(p4.east)+(0.6,0.0)$);
  \end{tikzpicture}
\end{figure}

Solving membership for levels of this hierarchy is a
longstanding open problem. Following Schützenberger's approach, it was first
investigated for languages of finite words. However, the question also
makes sense for more complex structures, in particular for the most natural
extension, infinite words. Schützenberger's result was first generalized to
infinite words by Perrin~\cite{pfo}, and a suitable algebraic framework for
languages of infinite words was set up by Wilke~\cite{womega}. Since a
regular language of infinite words is determined by regular languages
of finite words, finding a membership algorithm for languages of infinite
words does not usually require to start over. Instead these algorithms
are obtained by building on top of the algorithms for finite words, adding new
arguments, specific to infinite~words.

Regarding the hierarchy, membership is easily seen to be decidable for
$\sic1$. For \bsc1, the classical result of Simon~\cite{simon75} was
generalized from finite to infinite words by Perrin and
Pin~\cite{Perrin&Pin:Infinite-Words:2004:a}. For finite words,
membership to \sicd is  known to be decidable~\cite{arfi87,pwdelta},
a result lifted to infinite words in~\cite{dksig2,bojsig2infinite}. Following these
results, the understanding of the hierarchy remained stuck for years
until the framework was extended to new and more general problems than
membership.

Rather than asking whether a language is definable in a fragment \Fs, these
problems ask what is the best \Fs-definable ``approximation'' of this language
(with respect to specific criteria). The simplest example is
\emph{\Fs-separation}, which takes \emph{two} regular languages as input and
asks whether there exists a third language \emph{definable in \Fs} that
contains the first language and is disjoint from the second. Separation is
more general than membership: asking whether a regular language is definable
in \Fs is the same as asking whether it can be \Fs-separated from its (also
regular) complement. A consequence is that deciding these more general
problems is usually more challenging than deciding membership. However, their
investigation in the setting of finite words has also been very rewarding. A
good illustration is the transfer result of~\cite{pzqalt}, which states that
for all $i$, decidability of separation for \sici entails decidability of
membership for $\sic{i+1}$. Combined with an algorithm for
\sicd-separation~\cite{pzqalt}, this proved that \sict has decidable
membership. This result was strengthened in~\cite{pseps3}, which shows that
\sict-separation is decidable as well, thus obtaining decidability of
membership for~$\sic4$. Finally, in~\cite{pzqalt}, it was shown that \bscd has
decidable membership by using a generalization of separation for \sicd and
analyzing an algorithm solving this generalization.

It remained open to know whether it was possible to generalize with the same
success this new approach to the setting of infinite words. This is the
investigation that we carry out in the paper. More precisely, we rely on the
crucial notion of \ichains, designed in~\cite{pzqalt} for presenting and
proving membership and separation algorithms for finite words. We generalize
this concept to infinite words and successfully use it to prove that the
following problems are decidable: $\sicd$-separation, $\sict$-separation, and
$\bscd$ membership. This demonstrates that \ichains remain a suitable
framework for presenting arguments in the setting of infinite words. On the
other hand, new issues specific to infinite words arise, for example, we were
not able to generalize the transfer result from \sici-separation to
\sic{i+1}-membership (as a consequence, membership for \sic4 remains open).
Note that a different proof for deciding \bscd-membership has been obtained
independently in \cite{DBLP:journals/corr/KufleitnerW15a}. It gives a
decidable characterization based on topology and algebra.

Note that, for each problem, we pre-compute some information
by using the corresponding algorithm designed
in~\cite{pzqalt,pseps3} for finite words. This means that the involved
algorithms from~\cite{pzqalt,pseps3} are used as subroutines of our
algorithms.

We now present the problems in depth in Section~\ref{sec:preliminaries}, and we
solve them in the rest of the paper. A detailed outline is provided at the
end of Section~\ref{sec:preliminaries}.

\section{Presentation of the Problem}
\label{sec:preliminaries}
In this section, we first define the quantifier alternation hierarchy of
first-order logic. Then, we present the membership problem and the separation problem.

\subsection{The Quantifier Alternation  Hierarchy of First-Order Logic}

\noindent
{\bf Words and \iWords.} For the whole paper, we assume that a finite alphabet~$A$ is fixed.
We denote by $A^+$ the set of all finite nonempty words, and
by~$A^\infty$ the set of all infinite words over~$A$. In the paper, we
use the following terminology: the term ``word'' means a
finite word, and the term ``\iword'' means an infinite~word.

If $u$ is a word and $v$ is a word (resp. an \iword), we denote by
$uv$ the word (resp. \iword) obtained by concatenating $u$ to the left
of $v$. If $u \in A^+$ is a word, we denote by $u^\infty$ the
\iword $uuuu\cdots$ obtained as the infinite concatenation of $u$ with
itself.
If $u \in A^+ \cup A^\infty$ is a word or an \iword, we denote by \content{u}
the \emph{alphabet} of $u$, \emph{i.e.}, the set of letters of $u$.
We call \emph{language} (resp.\ \emph{\ilang}) a subset of $A^+$ (resp.\ of
$A^\infty$), \emph{i.e.}, a language of finite words (resp. of \iwords).

In the paper we are interested in \emph{regular} languages and
\ilangs. Regular \ilangs are those that can be equivalently defined by
monadic second-order logic, finite Büchi automata or finite \isemis.
We work with the definition of regular \ilangs
in terms of \isemis, recalled in
Section~\ref{sec:semi}.

\smallskip
\noindent
{\bf First-Order Logic.} Any word or \iword can be viewed as a logical
structure made of a linearly ordered sequence of positions labeled over
the alphabet $A$ (finite for words and infinite for \iwords). In
first-order logic (\fo), one can quantify over these positions and use
the following predicates.
\begin{itemize}
\item for each $a \in A$, a unary predicate $P_a$ selecting all
  positions labeled with an $a$.
\item a binary predicate '$<$' interpreted as the (strict)
  linear order over the positions.
\end{itemize}
Since any \fo sentence may be interpreted both on words and \iwords,
each sentence $\varphi$ defines two objects: a language $L_+ = \{w \in
A^+ \mid w \models \varphi\}$ and an \ilang $L_\infty = \{w \in
A^\infty \mid w \models \varphi\}$. For example, the sentence $\exists
x \exists y\ (x < y \wedge P_a(y))$ defines the language $A^+a\cup A^+aA^+$ and the \ilang $A^+aA^\infty$.

Thus, we may associate two classes of objects with \fo: a
class of languages~(we speak of \fo over words) and a class of \ilangs
(we speak of \fo over~\iwords).

\smallskip
\noindent
{\bf Quantifier Alternation.} It is usual to classify \fo sentences by
counting the quantifier alternations inside their prenex normal form.
Set $i \in \nat$, a sentence is said to be \sic{i} (resp. \pic{i}) if
its prenex normal form has either
\begin{itemize}
\item \emph{exactly} $i -1$ quantifier
  alternations (\emph{i.e.}, exactly $i$ blocks of
  quantifiers) starting with an $\exists$
  (resp.\ $\forall$), or
\item \emph{strictly less} than $i -1$ quantifier
  alternations (\emph{i.e.}, strictly less than $i$ blocks).
\end{itemize}
For example, the sentence $\exists x_1 \forall x_2 \forall x_3 \exists
x_4 \ \varphi(x_1,x_2,x_3,x_4)$, with $\varphi$ quantifier-free, is
\sict. Note that in general, the negation of a \sici sentence is not a
\sici sentence (it is called a \pici sentence). Hence, it is also usual to
define \bsci sentences as those that are Boolean combinations of \sici
and \pici sentences.

As for full first-order logic, each level \sici, \pici or \bsci
defines two classes of objects: a class of languages and a class of
\ilangs. Therefore, we obtain two hierarchies: a hierarchy of classes
of languages and a hierarchy of classes of \ilangs.
Both hierarchies are strict (refer to the figure in the introduction).

\subsection{Decision Problems}

Our objective is to investigate the quantifier alternation hierarchy of
first-order logic over \iwords. We rely on two decision problems in order to
carry out this investigation: the membership problem and the separation
problem.
Both problems are parametrized by a level in the hierarchy and come
in two versions: a `language' one and an `\ilang' one. Given a level \Fs
in the hierarchy, the \emph{membership problem} for \Fs is as follows:

\medskip

\begin{center}
  \begin{tikzpicture}

    \coordinate (c1) at (2.0,1.0);

    \node[anchor=mid west] (p1) at (0.0,0.5) {\bf IN};
    \node[anchor=mid west] (p2) at (0.0,0.0) {\bf OUT};

    \node[anchor=mid west] (p3) at (1.05,0.5) {A regular language $L$\qquad\qquad};
    \node[anchor=mid west] (p4) at (1.05,0.0) {Is $L$ \Fs-definable ?};

    \node[draw,thick,rounded
    corners=3pt,rectangle,fit=(p1) (p2) (p3) (p4) (c1)] (m1) {};

    \node[draw,thick,rounded
    corners=3pt,fill=white,rectangle] at (m1.north) {Language Membership Problem};

    \begin{scope}[xshift=8cm]
      \coordinate (c1) at (2.0,1.0);

      \node[anchor=mid west] (p1) at (0.0,0.5) {\bf IN};
      \node[anchor=mid west] (p2) at (0.0,0.0) {\bf OUT};

      \node[anchor=mid west] (p3) at (1.05,0.5) {A regular \ilang $L$\qquad\qquad};
      \node[anchor=mid west] (p4) at (1.05,0.0) {Is $L$ \Fs-definable ?};

      \node[draw,thick,rounded
      corners=3pt,rectangle,fit=(p1) (p2) (p3) (p4) (c1)] (m1) {};

      \node[draw,thick,rounded
      corners=3pt,fill=white,rectangle] at (m1.north) {\iLang Membership
        Problem};
    \end{scope}
  \end{tikzpicture}
\end{center}

\medskip

The separation problem is more general. Given three languages or three \ilangs
$K,L_1,L_2$, we say that $K$ \emph{separates} $L_1$ from~$L_2$ if
$L_1 \subseteq K \text{ and } L_2 \cap K = \emptyset$. For $\Fs$ a level in
the hierarchy, we say that $L_1$ is \emph{$\Fs$-separable} from~$L_2$ if there
exists a language or \ilang that is definable in \Fs and separates $L_1$
from~$L_2$. Note that when $\Fs$ is not closed under complement (\emph{e.g.},
when $\Fs = \sici$ or $\Fs = \pici$), the definition is not symmetrical: $L_1$
may be $\Fs$-separable from $L_2$ while $L_2$ is not $\Fs$-separable from
$L_1$. The separation problem for $\Fs$ is as follows:

\medskip

\begin{center}
  \begin{tikzpicture}

    \coordinate (c1) at (2.0,1.0);

    \node[anchor=mid west,inner sep=0pt] (p1) at (0.0,0.5) {\bf IN};
    \node[anchor=mid west,inner sep=0pt] (p2) at (0.0,0.0) {\bf OUT};

    \node[anchor=mid west,inner sep=0pt] (p3) at (1.05,0.5) {Two regular languages $L_1,L_2$};
    \node[anchor=mid west,inner sep=0pt] (p4) at (1.05,0.0) {Is $L_1$ \Fs-separable from
      $L_2$ ?};

    \node[draw,thick,rounded
    corners=3pt,rectangle,fit=(p1) (p2) (p3) (p4) (c1)] (m1) {};

    \node[draw,thick,rounded
    corners=3pt,fill=white,rectangle] at (m1.north) {Language Separation
      Problem};

    \begin{scope}[xshift=8cm]
      \coordinate (c1) at (2.0,1.0);

      \node[anchor=mid west,inner sep=0pt] (p1) at (0.0,0.5) {\bf IN};
      \node[anchor=mid west,inner sep=0pt] (p2) at (0.0,0.0) {\bf OUT};

      \node[anchor=mid west,inner sep=0pt] (p3) at (1.05,0.5) {Two regular \ilangs $L_1,L_2$};
      \node[anchor=mid west,inner sep=0pt] (p4) at (1.05,0.0) {Is $L_1$ \Fs-separable from
        $L_2$ ?};

      \node[draw,thick,rounded
      corners=3pt,rectangle,fit=(p1) (p2) (p3) (p4) (c1)] (m1) {};

      \node[draw,thick,rounded
      corners=3pt,fill=white,rectangle] at (m1.north) {\iLang Separation
        Problem};
    \end{scope}
  \end{tikzpicture}
\end{center}
\medskip

An important remark is that membership  reduces to
separation: a regular language or \ilang is definable in \Fs
iff it is \Fs-separable from its (also regular) complement. This makes
separation a more general problem
than membership.

Both problems have been extensively studied in the literature. Indeed, it has
been observed that obtaining an algorithm for the membership or separation
problem associated to a particular level \Fs usually yields a deep insight on
\Fs. This is well illustrated by the most famous result of this kind,
Schützenberger's Theorem~\cite{sfo,mnpfo}, which yields a (language)
membership algorithm for \fo. The result was later generalized to \ilangs by
Perrin~\cite{pfo}. These results and the techniques used to obtain them
provide not only a way to decide whether a regular ($\omega$-)language is
\fo-definable, but also a generic method for constructing an \fo sentence when
the ($\omega$-)language is definable. Since these first results, many efforts have been
devoted for obtaining membership and separation algorithms for each level in
the hierarchy. An overview of the results is presented in the following
table (omitted levels are open in all cases).

\begin{figure}[h]
  \begin{minipage}[c]{.48\linewidth}
    \begin{center}
      {\bf \large Membership Problem}\\
      \begin{tabular}{|c|c|c|}
        \hline
        & Language & \iLang \\
        \hline
        \fo    & Solved~\cite{sfo,mnpfo} & Solved~\cite{pfo} \\
        \hline
        \sicu  & Solved~(Folklore) & Solved~(Folklore)\\
        \hline
        \bscu  & Solved~\cite{simon75} & Solved~\cite{Perrin&Pin:Infinite-Words:2004:a}\\
        \hline
        \sicd  & Solved~\cite{arfi87,pwdelta} & Solved~\cite{dksig2}\\
        \hline
        \bscd  & Solved~\cite{pzqalt} & {\bf Open}\\
        \hline
        \sict  & Solved~\cite{pzqalt} & {\bf Open}\\
        \hline
        \sic{4} & Solved~\cite{pseps3}  &{\bf Open}\\
        \hline
      \end{tabular}
    \end{center}
  \end{minipage}
  \begin{minipage}[c]{.48\linewidth}
    \begin{center}
      {\bf \large Separation Problem}\\
      \begin{tabular}{|c|c|c|}
        \hline
        & Language & \iLang \\
        \hline
        \fo    & Solved~\cite{pzfo} & Solved~\cite{pzfo} \\
        \hline
        \sicu  & Solved~(Folklore) & Solved~(Folklore)\\
        \hline
        \bscu  & Solved~\cite{martens,pvzmfcs13} & Solved~\cite{bsc1omega}\\
        \hline
        \sicd  & Solved~\cite{pzqalt} & {\bf Open} \\
        \hline
        \bscd  &  {\bf Open} & {\bf Open}\\
        \hline
        \sict  & Solved~\cite{pseps3} & {\bf Open}\\
        \hline
        \sic{4} &  {\bf Open}  &{\bf Open}\\
        \hline
      \end{tabular}
    \end{center}
  \end{minipage}
\end{figure}

Our main objective is to bridge the gap between what is
known for languages and what is known for \ilangs. More precisely, we
want to extend the recent results of~\cite{pzqalt} and~\cite{pseps3}
to the setting of \iwords, \emph{i.e.}, to obtain membership algorithms for
\bscd, \sict and \sic{4} as well as separation algorithms for \sicd
and \sict. We were able to obtain these algorithms for \sicd, \sict
and \bscd as we state in the following theorem (we leave the case of
\sic{4}-membership for \ilangs open, we will come back to this point
in the conclusion).

\begin{theorem} \label{thm:main}
  The following properties hold:
  \begin{enumerate}[label=$\alph*)$]
  \item the \ilang separation problem is decidable for \sicd.  \item the \ilang membership problem is decidable for \bscd.
  \item the \ilang separation problem is decidable for \sict.
  \end{enumerate}
\end{theorem}

Our proof of Theorem~\ref{thm:main} consists in three algorithms, one for each
item in the theorem. An important remark is that each of these three
algorithms depends upon an algorithm of~\cite{pzqalt} or~\cite{pseps3}
solving the corresponding problem for languages:
\begin{itemize}
\item We present all algorithms in a specific framework which is
  adapted from the one used in~\cite{pzqalt}. In particular, we
  reuse the key notion of ``\ichain'' (generalized to \iwords in
  a straightforward way).
\item We actually reuse the language algorithms of~\cite{pzqalt}
  and~\cite{pseps3} as subprocedures in our algorithms for \ilangs.
\end{itemize}

The remainder of the paper is devoted to proving
Theorem~\ref{thm:main}. In Section~\ref{sec:semi},
we recall classical notions required for our definitions
and proofs: the \isemi definition of regular \ilangs  and logical
preorders. In Section~\ref{sec:chains}, we present the general
framework used in the paper. In particular, we introduce a
notion that will be at the core of all our algorithms: ``\ichains''
(which are adapted and reused from~\cite{pzqalt}). We then devote a
section to each algorithm: Section~\ref{sec:sig2} to
\sicd-separation, Section~\ref{sec:bsig2} to \bscd-membership and
Section~\ref{sec:sig3} to \sict-separation.

\def\texorpdfstring#1#2{#1}

\section{\texorpdfstring{Preliminaries}{Preliminaries}}
\label{sec:semi}
In this section, we recall some classical notions that we will need.
First, we present the definition of regular \ilangs in terms of
\isemis. Then, we define the logical preorders that one may associate
to each level \sici in the hierarchy.

\subsection{Semigroups and \iSemis}

We briefly recall the definition of regular languages and \ilangs in
terms of semigroups and \isemis. For details,
see~\cite{Perrin&Pin:Infinite-Words:2004:a}.

\smallskip
\noindent {\bf Semigroups.} A semigroup is a set $S$ equipped with an
associative operation $s \cdot t$ (often written $st$). In particular, $A^+$
equipped with concatenation is a~semigroup. Given a \emph{finite} semigroup
$S$, it is easy to see that there is an integer $\omega(S)$ (denoted by
$\omega$ when $S$ is understood) such that for all $s$ of $S$, $s^\omega$ is
idempotent:~$s^\omega = s^\omega s^\omega$.

Given a language $L$ and a morphism $\alpha: A^+ \to S$, we say that
$L$ is \emph{recognized} by $\alpha$ if and only if there exists $F
\subseteq S$ such that $L = \alpha^{-1}(F)$. It is well-known that a
language is regular if and only if it may be recognized by a
\emph{finite} semigroup.

\smallskip
\noindent
{\bf\iSemis.} An \emph{\isemi} is a pair
$(S_+,S_\infty)$, where $S_+$ is a semigroup and $S_\infty$ is a set.
Moreover, $(S_+,S_\infty)$ is equipped with two additional products: a
\emph{mixed product} $S_+ \times S_\infty \rightarrow S_\infty$
mapping $s,t \in S_+,S_\infty$ to an element $st$ of $S_\infty$, and
an \emph{infinite product} $(S_+)^\infty \rightarrow S_\infty$ mapping
an infinite sequence $s_1,s_2,\dots \in (S_+)^\infty$ to an element
$s_1s_2\cdots$ of $S_\infty$. We require these products to satisfy all
possible forms of associativity. For $s\in S_+$, we let $s^\infty$ be
the infinite product $sss\cdots\in S_\infty$. Note that
$(A^+,A^\infty)$ is an \isemi. See~\cite{Perrin&Pin:Infinite-Words:2004:a} for
further details.

We say that $(S_+,S_\infty)$ is \emph{finite} if both $S_+$ and
$S_\infty$ are. Note that even if an \isemi is finite, it is not clear
how to represent the infinite product, since the set of infinite
sequences of $S_+$ is uncountable. However, it has been shown by
Wilke~\cite{womega}  that the infinite product is fully determined by
the mapping $s \mapsto s^\infty$. This makes it possible to finitely
represent any finite \isemi.

Morphisms of \isemis are defined in the natural way. In particular, observe
that any \isemi morphism $\alpha: (A^+,A^\infty) \to (S_+,S_\infty)$ defines
two maps: a semigroup morphism $\alpha_+: A^+ \to S_+$ and a map
$\alpha_\infty: A^\infty \to S_\infty$ (when there is no ambiguity, we shall
write $\alpha(w)$ to mean $\alpha_+(w)$ if $w\in A^+$ or $\alpha_\infty(w)$ if
$w\in A^\infty$). Therefore, a morphism recognizes both languages (the
languages $\alpha_+^{-1}(F_+)$ for $F_+ \subseteq S_+$) and \ilangs (the
\ilangs $\alpha_\infty^{-1}(F_\infty)$ for $F_\infty \subseteq S_\infty$). An
\ilang is regular iff it may be recognized by a morphism into a \emph{finite}
\isemi.

\smallskip
\noindent
{\bf Syntactic Morphisms.} It is known that given any regular
language (resp. \ilang) $L$ there exists a canonical morphism
$\alpha_L: A^+ \to S$ (resp. $\alpha_L: (A^+,A^\infty) \to
(S_+,S_\infty)$) recognizing $L$. This object is called the
syntactic morphism of $L$. We refer
the reader to~\cite{Perrin&Pin:Infinite-Words:2004:a} for the
detailed definition of this object. In the paper we only use
two properties of the syntactic morphism. The first is simply that
given any regular \ilang $L$, one may compute its syntactic
morphism from any representation of $L$. We state the second
one below.
\begin{fact} \label{fct:synta}
  Let $i \geq 1$ and let $L$ be a regular \ilang. Then $L$ is definable
  in \bsci iff so are all languages and \ilangs recognized by its syntactic
  morphism.
\end{fact}

The proof of Fact~\ref{fct:synta} may be found
in~\cite{Perrin&Pin:Infinite-Words:2004:a} (in fact, this holds for any class
of \ilangs which is a ``variety'' of \ilangs, not just for \bsci). In view~of~this,  the syntactic morphism is central for membership questions:
deciding if~a language is definable in \bsci amounts to deciding a property of
its syntactic morphism. This is the approach used in our \bscd-membership
algorithm (see Section~\ref{sec:bsig2}).

\medskip
\noindent
{\bf Morphisms and Separation.}  When working on separation, we are
given two input languages or \ilangs. It is convenient to consider a
single recognizing object for both inputs rather than two separate objects.
This is not restrictive: given two languages (resp. two \ilangs) and
two associated recognizing morphisms, one can define and compute a
single morphism that recognizes them both. For example, if $L_0 \subseteq A^\infty$ is recognized by $\alpha_0:
(A^+,A^\infty) \to (S_+,S_\infty)$ and $L_1 \subseteq A^\infty$  by $\alpha_1: (A^+,A^\infty) \to (T_+,T_\infty)$, then
$L_0$ and $L_1$ are both recognized by
$\alpha: (A^+,A^\infty) \to     (S_+\times T_+,S_\infty \times    T_\infty)$
with
$\alpha(w) =(\alpha_0(w),\alpha_1(w))$.

\medskip
\noindent
{\bf Alphabet Compatible Morphisms.} It
will be convenient to work with morphisms that satisfy an additional
property. A morphism $\alpha: (A^+,A^\infty) \rightarrow
(S_+,S_\infty)$ is said to be \emph{alphabet compatible} if for all
$u,v \in A^+ \cup A^\infty$, $\alpha(u) = \alpha(v) $ implies
$\content{u} = \content{v}$. Note that when $\alpha$ is alphabet
compatible, for all $s \in S_+ \cup S_\infty$, $\content{s}$ is
well defined as the unique $B \subseteq A$ such that for all $u
\in \alpha^{-1}(s)$, we have $\content{u} = B$ (if $s$ has no preimage
then we simply set $\content{s} = \emptyset$).

To any morphism $\alpha: (A^+,A^\infty) \rightarrow
(S_+,S_\infty)$, we associate a morphism $\beta$, called the
\emph{alphabet completion} of $\alpha$. The morphism $\beta$
recognizes all \ilangs recognized by $\alpha$ and is alphabet
compatible. If $\alpha$ is already alphabet compatible, then
$\beta = \alpha$. Otherwise, observe that $2^A$ is a semigroup
with union as the multiplication and $(2^A,2^A)$ is therefore
an \isemi. Hence, we can define $\beta$ as the morphism:
$\beta: (A^+,A^\infty)  \to      (S_+ \times 2^A,S_\infty \times 2^A)$ with $\beta(w)   =  (\alpha(w),\content{w})$.

\subsection{Logical Preorders}

To each level \sici in the hierarchy, one may associate preorders on
the sets of words and \iwords. The definition is based on the notion
of quantifier rank. The \emph{quantifier rank} of a first-order
formula is the length of the longest sequence of nested quantifiers
inside the formula. For example, the following formula,
\[
  \exists x\ P_b(x) \wedge \neg (\exists y\ (y < x \wedge P_c(y)) \wedge
  (\forall y\exists z\ x < y < z \wedge P_b(y)))
\]
\noindent
has quantifier rank $3$. It is well-known (and easy to show) that for a fixed $k$, there
is a finite number of non-equivalent first-order formulas of rank less
than $k$.

We may now define the preorders. Note that while we  define
two preorders for each level \sici (one on $A^+$ and one on
$A^\infty$), we actually use the same notation for both. Set $i \geq
1$ as a level in the hierarchy and $k \geq 1$ as a quantifier rank.
Given two words $w,w' \in A^+$ (resp two \iwords $w,w' \in A^\infty$),
we write $w \ksieq{i} w'$ if and only if \emph{any} \sici formula of
rank at most $k$ that is satisfied by $w$ is satisfied by $w'$ as
well. One may verify that $\ksieq{i}$ is preorder. Moreover, it is
immediate from the definition that the preorders get refined
when $k$ increases: $w \sieq{k+1}{i} w' \Rightarrow w \sieq{k}{i} w'$.

Denote by $\kbceq{i}$ the equivalence generated by $\ksieq{i}$:
$w \kbceq{i} w'$ when $w \ksieq{i} w'$ and $w' \ksieq{i} w$. That is, $w \kbceq{i} w'$ if and only if $w,w'$ satisfy the same \sici
sentences (or equivalently the same \bsci sentences, which are just
Boolean combinations of \sici sentences). The following fact may be
verified from the~definition.
\begin{fact} \label{fct:refine}
  Let $k,i \geq 1$ and let $u,v$ be two words or two \iwords, then
  \[
    (1)\ u \sieq{k+1}{i} v  \Rightarrow  u \sieq{k}{i} v,
    \qquad(2)\ u \bceq{k+1}{i} v \Rightarrow  u \bceq{k}{i} v
    \qquad(3)\ u \ksieq{i+1} v    \Rightarrow  u \kbceq{i} v.
  \]
\end{fact}
We finish the section with a few properties about the preorders
\ksieq{i}. The proofs are easy and omitted (they are obtained
with standard \efgame arguments).  We start with decomposition and
composition lemmas.

\begin{lemma}[Decomposition Lemma]
  \label{lem:EFdecomp}
  Let $i,k \geq 1$ and let $u,v$ be two words or two \iwords such that
  $u \ksieq{i} v$. Then for any decomposition $u=u_1u_2$ of $u$, there
  exist $v_1,v_2$ such that $v=v_1v_2$, $u_1\sieq{k-1}{i} v_1$ and
  $u_2 \sieq{k-1}{i} v_2$ .
\end{lemma}

\begin{lemma}[Composition Lemma]
  \label{lem:EFmult}
  Let $i,k \geq 1$, let $u_1,v_1$ be two words such that $u_1\ksieq{i}
  v_1$, and $u_2,v_2$ be either two words or two \iwords such that $u_2
  \ksieq{i} v_2$. Then  $u_1u_2\ksieq{i} v_1v_2$ and $u_1^\infty
  \ksieq{i} v_1^\infty$.
\end{lemma}
The last composition that we state is specific to \iwords.

\begin{lemma}
  \label{lem:EF}
  Let $i,k \geq 1$, $u \in A^+$ be a word and $v \in A^\infty$ be an
  \iword such that $v \ksieq{i} u^\infty$. Then for any $\ell \geqslant
  2^k$, we have $u^\infty \ksieq{i+1} u^\ell v$.
\end{lemma}

In particular we will use the special case of Lemma~\ref{lem:EF} in
which $i=1$. In this case, one can verify that given $u \in A^+$ and
$v \in A^\infty$, when $\content{u} = \content{v}$, we have $v \ksieq{1}
u^\infty$ for any $k \geq 1$. Hence we have the following corollary of
Lemma~\ref{lem:EF}.

\begin{corollary}
  \label{cor:EF}
  Let $k \geq 1$, $u \in A^+$ be a word and let $v \in A^\infty$ be an
  \iword such that $\content{u} = \content{v}$. Then for any $\ell
  \geqslant 2^k$, we have $u^\infty \ksieq{2} u^\ell v$.
\end{corollary}

\section{\texorpdfstring{\iChains for \iLangs}{\Chains for Omega-Languages}}
\label{sec:chains}
As explained, all algorithms for \ilangs of this paper are
strongly related to the algorithms for languages of~\cite{pzqalt}
and~\cite{pseps3}. In particular, we adapt and reuse the key notion of
``\ichain'' which was introduced in~\cite{pzqalt}. The section
is devoted to the presentation of this notion. First, we define
\ichains. We then detail the link between \ichains and our
decision problems, first for \sici, then for \bsci.

\subsection{\iChains}

\iChains were initially introduced in~\cite{pzqalt} as a tool designed
to investigate the (language) separation problem for the logics \sici
and \bsci. A ``set of \ichains'' can be associated to any morphism
$\alpha: A^+ \to S$ into a finite semigroup $S$. Intuitively, this set
captures information about what \sici and \bsci can express about the
languages recognized by $\alpha$ (including which ones are separable
with \sici and \bsci). The definition is based on the following lemma.

\begin{lemma} \label{lem:preo}
  Let $i,k \geq 1$ and let $L_1,L_2$ be two languages or two \ilangs.
  Then
  $L_1$ is {\bf not} \sici-separable (resp.\ {\bf not} \bsci-separable) from $L_2$ iff for all $k
  \geq 1$, there exist $w_1 \in L_1$ and $w_2 \in L_2$ such that
  $w_1 \ksieq{i} w_2$ (resp.\ such that
  $w_1 \kbceq{i} w_2$).
\end{lemma}

\begin{proof}
  We prove the first item, the second one is obtained similarly. Assume
  first that $L_1$ is {\bf not} \sici-separable from $L_2$. Set $k \geq
  1$ and consider the language or \ilang $K = \{w \mid \exists v \in L_1
  \text{   s.t. } v \ksieq{i} w\}$. By definition, $L_1 \subseteq K$ and
  $K$ may be defined by a \sici sentence of rank $k$ (this is because
  there are finitely many nonequivalent sentences of rank $k$). Therefore, there
  exists $w_2 \in L_2 \cap K$ (otherwise $K$ would separate $L_1$ from
  $L_2$, which is impossible by hypothesis). By definition of $K$, we get some
  $w_1 \in L_1$ such that $w_1 \ksieq{i} w_2$ which terminates the proof
  of this direction.

  For the other direction, assume that for all $k \geq 1$, there exist
  $w_1 \in L_1$ and $w_2 \in L_2$ such that $w_1 \ksieq{i} w_2$ and let
  $K \supseteq L_1$ be a language or \ilang that is defined by a \sici sentence
  $\varphi$. We prove that $K \cap L_2 \neq \emptyset$, \emph{i.e.}, that $K$
  cannot be a separator. Let $k$ be the rank of $\varphi$, we obtain
  $w_1 \in L_1$ and $w_2 \in L_2$ such that $w_1 \ksieq{i} w_2$ from our
  hypothesis. Since $L_1\subseteq K$, we have $w_1\in K$, and by definition of
  $K$, $w_1 \models \varphi$. Since $\varphi$ is of rank~$k$, by definition of
  $\ksieq{i}$ we must have
  $w_2 \models \varphi$, \emph{i.e.}, $w_2 \in L_2 \cap K$ which terminates the
  proof.
\end{proof}

Lemma~\ref{lem:preo} states simple criteria equivalent to \sici-
and \bsci-separability. However, both criteria involve a
quantification over all natural numbers. Therefore, it is not immediate that
they can be decided. Indeed, since both $A^+$ and $A^\infty$ are
infinite sets, $\ksieq{i}$ and $\kbceq{i}$ are endlessly refined as
$k$ gets larger.

\iChains are designed to deal with this issue. The separation problem takes
two \emph{regular} languages or \ilangs as input. Therefore, we have a single
morphism that recognizes them both. For example, in the \ilang case, we have
$\alpha: (A^+,A^\infty) \to (S_+,S_\infty)$ (with $(S_+,S_\infty)$ a finite
\isemi) that recognizes both inputs. Intuitively, $S_+$ and $S_\infty$ are
finite abstractions of $A^+$ and~$A^\infty$. Therefore, we may abstract the
preorders $\ksieq{i}$ on these two finite sets: this is what \ichains are. For
example, we say that $(s,t) \in (S_\infty)^2$ is a \ichain (of length $2$) for
$\alpha$ iff for all $k$, there exist $u,v \in A^\infty$ such that
$\alpha(u) = s$, $\alpha(v) = t$ and $u \ksieq{i} v$. For \ilangs recognized
by $\alpha$, it is then easy to adapt the two criteria of Lemma~\ref{lem:preo}
to work directly with the \ichains associated to $\alpha$. In other words, we
reduce separation to the (still difficult) problem of computing the set of
\ichains associated to a given input morphism.

\medskip
\noindent {\bf \Chains.} Let us now define \chains. Given a finite set $S$, a
\emph{\chain over $S$} is simply a finite word over $S$ (\emph{i.e.}, an
element of $S^+$). We shall only consider \chains over $S_+$ and over
$S_\infty$, where $S_+$ and $S_\infty$ are the two components of some \isemi
$(S_+,S_\infty)$. A remark about notation is in order: a word is usually
denoted as the concatenation of its letters. However, since $S_+$ is a
semigroup, this would be ambiguous: when $st \in (S_+)^+$, $st$ could
either mean a word with 2 letters $s$ and $t$, or the product of $s$ and $t$
in $S_+$. To avoid confusion, we will write $(s_1,\dots,s_n)$ a \chain of
length $n$. We denote \chains by $\bar{s},\bar{t},\dots$ and
sets of \chains by $\Ss,\Ts,\dots$.

If $(S_+,S_\infty)$ is an \isemi, then for all $n \in \nat$,
$(S_+)^n$ is a semigroup when equipped with the componentwise
multiplication $(s_1,\dots,s_n)(t_1,\dots,t_n)=(s_1t_1,\dots,s_nt_n)$. Moreover, the pair $((S_+)^n,(S_\infty)^n)$ is an
\isemi.

\medskip
\noindent
{\bf \iChains.} Fix $i \geq 1$ and $x \in \{+,\infty\}$. We associate
a set of \ichains to any map $\beta: A^x \to S$ where $S$ is a finite
set. The set $\Cs_i[\beta] \subseteq S^+$ of \ichains for $\beta$ is
defined as follows. Let $\bar{s} = (s_1,\dots,s_n) \in S^+$ be a
\chain. We have $\bar{s} \in \Cs_i[\beta]$ if and only if for all $k
\in \nat$, there exist $w_1,\dots,w_n \in A^x$ such that:
\[
  w_1 \ksieq{i} w_2 \ksieq{i} \cdots \ksieq{i} w_n  \text{ and for all } j,\
  \beta(w_j) = s_j.
\]
Moreover, we denote by $\Cs_{i,n}[\beta]$ the restriction of this
set to \chains of length $n$ only (\emph{i.e.}, $\Cs_{i,n}[\beta] =
\Cs_{i}[\beta] \cap S^n$).

\medskip
\noindent
{\bf \iChains for an \iSemi Morphism.} It follows from the
definition of \ichains that one may associate a set  $\Cs_i[\alpha]$
to any semigroup morphism $\alpha: A^+ \to S$. This set is exactly
the set of \ichains associated to $\alpha$ as~defined~in~\cite{pzqalt}.

Moreover, given a morphism $\alpha: (A^+,A^\infty) \to (S_+,S_\infty)$
into a finite \isemi $(S_+,S_\infty)$, one may associate two sets of
\ichains to $\alpha$: one to the morphism $\alpha_+: A^+ \to S_+$
($\Cs_{i}[\alpha_+] \subseteq (S_+)^+$) and one to the map
$\alpha_\infty: A^\infty \to S_\infty$ ($\Cs_{i}[\alpha_\infty]
\subseteq (S_\infty)^+$). We may now link \ichains to the separation
problem.

\subsection{\iChains and Separation for \sici}

We now connect \ichains to the separation problem. We begin with
the simplest connection, which is between \ichains of length $2$ and
separation for~\sici.

\begin{theorem} \label{thm:chainsep}
  Let $i \geq 1$, $x \in \{+,\infty\}$ and $\beta: A^x \to S$ a map into
  a finite set $S$. Given $F_1,F_2 \subseteq S$, $L_1 = \beta^{-1}(F_1)$
  and $L_2 = \beta^{-1}(F_2)$, the following are equivalent
  \begin{enumerate}
  \item $L_1$ is {\bf not} \sici-separable from $L_2$.
  \item there exist $s_1 \in F_1$ and $s_2 \in F_2$
    such that $(s_1,s_2) \in \Cs_{i,2}[\beta]$.
  \end{enumerate}
\end{theorem}
Theorem~\ref{thm:chainsep} is a straightforward consequence
Lemma~\ref{lem:preo} (statement for \sici). In view of the theorem, our
approach for the \sici-separation problem is as follows:
\begin{itemize}
\item for languages, we look for an algorithm computing
  $\Cs_{i,2}[\alpha]$ from an input morphism $\alpha: A^+ \to S$
  into a finite semigroup $S$.
\item for \ilangs, we look for an algorithm computing
  $\Cs_{i,2}[\alpha_\infty]$ from an input morphism $\alpha:
  (A^+,A^\infty) \to (S_+,S_\infty)$  into a finite \isemi
  $(S_+,S_\infty)$. In particular, this algorithm typically
  involves computing $\Cs_{i,2}[\alpha_+]$ first, which can be
  achieved by reusing first item, \emph{i.e.}, the  algorithm for word languages.
\end{itemize}
This approach is exactly the one from~\cite{pzqalt,pseps3} to solve separation for \sicd and
\sict over finite words: the following theorems are proven in these papers.

\begin{theorem}[\cite{pzqalt}] \label{thm:fisicd}
  Given as input a morphism $\alpha: A^+ \to S$ into a finite semigroup
  $S$, one can compute the set $\Cs_{2,2}[\alpha]$ of \dchains of length
  $2$ for $\alpha$.
\end{theorem}

\begin{theorem}[\cite{pseps3}] \label{thm:fisict}
  Given as input a morphism $\alpha: A^+ \to S$ into a finite semigroup
  $S$, one can compute the set $\Cs_{3,2}[\alpha]$ of \tchains of length
  $2$ for $\alpha$.
\end{theorem}

We generalize these two theorems in Section~\ref{sec:sig2} (for \sicd) and
Section~\ref{sec:sig3} (for \sict) for \iwords by presenting two
new algorithms. These algorithms both take a morphism
$\alpha: (A^+,A^\infty) \to (S_+,S_\infty)$ as input and compute the sets
$\Cs_{2,2}[\alpha_\infty]$ and $\Cs_{3,2}[\alpha_\infty]$ respectively. Note
that the algorithms of Theorem~\ref{thm:fisicd} and Theorem~\ref{thm:fisict}
are reused as sub-procedures in these new algorithms for \ilangs: computing
$\Cs_{2,2}[\alpha_\infty]$ and $\Cs_{3,2}[\alpha_\infty]$ requires to first
compute $\Cs_{2,2}[\alpha_+]$ and $\Cs_{3,2}[\alpha_+]$ respectively.

\begin{remark}
  The algorithms of
  Theorems~\ref{thm:fisicd} and~\ref{thm:fisict} both work with
  objects that are actually more general than \ichains: the \sicd
  algorithm works with ``\djuns'' and the \sict algorithm with an even
  more general notion: ``$\Sigma_{2,3}$-trees''.
  We do not present these more general notions because we do
  not need them outside of the algorithms of
  Theorems~\ref{thm:fisicd} and~\ref{thm:fisict}, which we use as black boxes.
\end{remark}

\subsection{\iChains and Separation for \bsci}

We finish by presenting the connection between the separation problem
for \bsci and \ichains. This time, the connection depends on the whole
set of \ichains. More precisely, it depends on yet another notion called
\emph{alternation}.

Set $x \in \{+,\infty\}$ and $\beta: A^x \to S$ as a map into a
finite set $S$. We say that a pair $(s,t) \in S^2$ is
$\sici$-alternating for $\beta$ iff for all $n \geq 1$,
$(s,t)^n \in \Cs_i[\beta]$ (where by $(s,t)^n$, we mean the \chain
$(s,t,s,t,\dots,s,t)$ of length $2n$).

\begin{theorem} \label{thm:bcsep}
  Let $i \geq 1$, $x \in \{+,\infty\}$ and $\beta: A^x \to S$ a map into
  a finite set $S$. Given $F_1,F_2 \subseteq S$, $L_1 = \beta^{-1}(F_1)$
  and $L_2 = \beta^{-1}(F_2)$, the following are equivalent,
  \begin{enumerate}
  \item $L_1$ is {\bf not} \bsci-separable from $L_2$.
  \item there exist $s_1 \in F_1$ and $s_2 \in F_2$
    such that $(s_1,s_2)$ is \sici-alternating.
  \end{enumerate}
\end{theorem}

\begin{proof}
  There are two directions to prove. Assume first that $L_1$ is not
  \bsci-separable from $L_2$. By Lemma~\ref{lem:preo}, we know that
  for all $k \geq 1$ we have $w_1 \in L_1$ and $w_2 \in L_2$ such that
  $w_1 \kbceq{i} w_2$. Hence, it follows that for all $k \geq 1$, we
  have $w_1 \in L_1$ and $w_2 \in L_2$ such that,
  \[
    w_1 \ksieq{i} w_2 \ksieq{i} w_1 \ksieq{i} w_2 \ksieq{i} \cdots
  \]
  Note that $w_1,w_2$ depend on $k$, but since $S$
  is finite, one can assume their images to be constant for infinitely many
  values of $k$. Since $w_1 \sieq{k+1}{i} w_2 \Rightarrow w_1 \sieq{k}{i} w_2$ (see
  Fact~\ref{fct:refine}), we may assume that they are constant for all $k$,
  \emph{i.e.}, that there exist $s_1 \in
  \beta(L_1)$ and $s_2 \in \beta(L_2)$ such that for all $k \geq 1$
  the corresponding $w_1$ and $w_2$ are mapped to $s_1$ and $s_2$ respectively. This
  exactly means thats for all $n \geq 1$, the \chain $(s_1,s_2)^n$ is a \ichain for
  $\beta$. Therefore, $(s_1,s_2)$ is \sici-alternating, which terminates
  the proof of this direction.

  It remains to prove the other direction. Assume that there exist $s_1
  \in \beta(L_1)$ and $s_2 \in \beta(L_2)$ such that $(s_1,s_2)$ is
  \sici-alternating. We have to prove that $L_1$ is {\bf not}
  \bsci-separable from $L_2$. We know from Lemma~\ref{lem:preo} that it
  suffices to prove that for all $k \geq 1$, there exist $w_1 \in L_1$ and
  $w_2 \in L_2$ such that $w_1 \kbceq{i} w_2$.

  Set $k \geq 1$. Since there are only finitely many nonequivalent \bsci formulas of
  rank $k$, the relation $\kbceq{i}$ has finite index.  Let $\ell$
  be the number of equivalence classes of $\kbceq{i}$. Since $(s_1,s_2)$ is
  \sici-alternating, we know that the \chain $(s_1,s_2)^{\ell}$ of
  length $2\ell$ is a \ichain for $\beta$. Hence, we have $2\ell$ words
  $u_1,\dots,u_{2\ell}$ such that for all $j \geq 1$, $u_{2j-1} \in L_1$
  and $u_{2j} \in L_2$, and $u_1 \ksieq{i} u_2 \ksieq{i} \cdots \ksieq{i}
  u_{2\ell}$. By choice of $\ell$, the pigeonhole principle
  gives $j < h$ such that $u_j \kbceq{i} u_h$. Since
  $u_j  \ksieq{i} u_{j+1} \ksieq{i} u_{h}\ksieq{i} u_j$, it follows that $u_j
  \kbceq{i} u_{j+1}$ which terminates the proof since either $u_j \in
  L_1$  and $u_{j+1} \in L_2$, or $u_j \in L_2$  and $u_{j+1} \in L_1$.
\end{proof}

In view of Theorem~\ref{thm:bcsep}, the separation problem for \bsci reduces to
the computation of the $\sici$-alternating pairs. Unfortunately,
whether there exists such an algorithm is open for $i \geq 2$, even in
the case of finite words.

However, Theorem~\ref{thm:bcsep} yields an immediate corollary that
applies to membership only. Given $x \in \{+,\infty\}$ and $\beta: A^x
\to S$ a map into a finite set $S$, we say that $\beta$ has
\emph{bounded \sici-alternation} iff there exists no
$\sici$-alternating pair $(s,t) \in S^2$ for $\beta$ such that $s \neq
t$.

\begin{corollary} \label{cor:bcmemb}
  Let $i \geq 1$, $x \in \{+,\infty\}$ and $\beta: A^x \to S$ be a map into
  a finite set~$S$. Then for any $F \subseteq S$, $\beta^{-1}(F)$ is
  \bsci-definable if and only if $\beta$ has bounded \sici-alternation.
\end{corollary}

\noindent
Combining Corollary~\ref{cor:bcmemb} with Fact~\ref{fct:synta} yields
a criterion for \bsci-membership: a regular language (resp.
\ilang) is definable in \bsci iff its syntactic morphism has bounded
\sici-alternation. This is the approach used in~\cite{pzqalt}
to obtain a (language) membership algorithm for \bscd. More precisely,
the following result~is~proved.

\begin{theorem}[\cite{pzqalt}] \label{thm:fibscd}
  Given as input a morphism $\alpha: A^+ \to S$ into a finite semigroup
  $S$, one can decide whether $\alpha$ has bounded \sicd-alternation or
  not.
\end{theorem}

\noindent
In Section~\ref{sec:bsig2} we obtain our (\ilang) algorithm for
\bscd-membership by~proving that given a morphism
$\alpha: (A^+,A^\infty) \to (S_+,S_\infty)$ as input, one can decide whether
$\alpha_\infty$ has bounded \sicd-alternation or not. More precisely, we prove
that $\alpha_\infty$ having bounded \sicd-alternation is equivalent to two
decidable properties of $\alpha$. The first one is that $\alpha_+$ has bounded
\sicd-alternation (which we can decide by Theorem~\ref{thm:fibscd}) and the
second is a simple equation that needs to be satisfied by~$(S_+,S_\infty)$.

\section{\texorpdfstring{A Separation Algorithm for \sicd}{A Separation Algorithm for  Sigma2}}
\label{sec:sig2}
In this section, we present an algorithm for the \ilang separation
problem associated to \sicd. As expected, this algorithm is based on
the computation of \dchains of length $2$ (see
Theorem~\ref{thm:chainsep}): we prove that given a morphism $\alpha$
into a finite \isemi, one can compute $\Cs_{2,2}[\alpha_\infty]$.

Given any \emph{alphabet compatible} morphism $\alpha: (A^+,A^\infty)
\to (S_+,S_\infty)$ into a finite \isemi, we denote by
$\Sat_{\sicd}(\alpha)$  the set of all pairs
\[
  (r_1(s_1)^\infty,\;r_2(s_2)^\omega t_2) \in S_\infty\times S_\infty
\]
with $(r_1,r_2)  \in \Cs_{2,2}[\alpha_+]$, $(s_1,s_2) \in
\Cs_{2,2}[\alpha_+]$, $t_2 \in \alpha(A^\infty)$ and $\content{s_1} =
\content{t_2}$ (note that the last condition is well defined since
$\alpha$ is alphabet compatible).

\begin{proposition} \label{prop:sig2main}
  Let $\alpha: (A^+,A^\infty) \to (S_+,S_\infty)$ be an alphabet
  compatible morphism into a finite \isemi $(S_+,S_\infty)$. Then,
  $\Cs_{2,2}[\alpha_\infty] = \Sat_{\sicd}(\alpha)$.
\end{proposition}

A simple consequence of Proposition~\ref{prop:sig2main} is that the
\ilang separation problem is decidable for \sicd. Indeed, recall
that for any two regular \ilangs, one may compute a single alphabet
compatible \isemi morphism that recognizes them both.  Therefore, it follows
from Theorem~\ref{thm:chainsep} that deciding \sicd-separation amounts
to having an algorithm that computes $\Cs_{2,2}[\alpha_\infty]$ from
$\alpha$.

We obtain this algorithm from Proposition~\ref{prop:sig2main} since
$\Sat_{\sicd}(\alpha)$ may be computed, given $\alpha$ as input. Indeed, by
Theorem~\ref{thm:fisicd}, we already know that the set $\Cs_{2,2}[\alpha_+]$
can be computed from $\alpha$. Once one has this set in hand, computing
$\Sat_{\sicd}(\alpha)$ is a simple matter. Hence, we obtain the desired
corollary.

\begin{corollary} \label{cor:sig2sep}
  The \ilang separation problem is decidable for \sicd.
\end{corollary}

An important remark is that we use Theorem~\ref{thm:fisicd} as a
black box: we do not reprove that $\Cs_{2,2}[\alpha_+]$ may be
computed from $\alpha_+$. This is not an immediate result. In fact the
proof of~\cite{pzqalt} requires to use a framework that is more
general than \dchains (that of ``\djuns'') as well as arguments that are
independent from those that we are going to use to prove
Proposition~\ref{prop:sig2main}.

It remains to prove each inclusion of Proposition~\ref{prop:sig2main}.
We first prove the easiest one: $\Cs_{2,2}[\alpha_\infty] \supseteq
\Sat_{\sicd}(\alpha)$ (this proves correctness: all
computed \chains  are indeed \dchains). 

\subsection{Correctness Proof in Proposition~\ref{prop:sig2main}}

In this subsection, we show 
${\Cs_{2,2}[\alpha_\infty] \supseteq \Sat_{\sicd}(\alpha)}$. Set
$(r_1,r_2) \in \Cs_{2,2}[\alpha_+]$, $(s_1,s_2) \in \Cs_{2,2}[\alpha_+]$ and
$t_2 \in \alpha(A^\infty)$ such that $\content{s_1} = \content{t_2}$. Our
objective is to prove that
$(r_1(s_1)^\infty,r_2(s_2)^\omega t_2) \in \Cs_{2,2}[\alpha_\infty]$. Set
$k \geq 1$. By definition, we need to find two \iwords $w_1 \ksieq{2} w_2$
such that $\alpha(w_1) = r_1(s_1)^\infty$ and
$\alpha(w_2) = r_2(s_2)^\omega t_2$.

By hypothesis, we have four words $x_1,x_2,y_1,y_2 \in A^+$ such that
$x_1 \ksieq{2} x_2$, $y_1 \ksieq{2} y_2$, $\alpha(x_1) = r_1$,
$\alpha(x_2) = r_2$, $\alpha(y_1) = s_1$ and $\alpha(y_2) = s_2$.
Moreover, we have an \iword $v \in A^\infty$ such $\alpha(v) = t_2$
and $\content{y_1} = \content{v}$. Set $w_1 = x_1(y_1)^\infty$ and
$w_2 = x_2(y_2)^{2^k\omega}v$. Observe that by definition, we have
$\alpha(w_1) = r_1(s_1)^\infty$ and $\alpha(w_2) = r_2(s_2)^\omega
t_2$. Therefore, it remains to prove that $w_1 \ksieq{2} w_2$.

By Corollary~\ref{cor:EF}, we obtain that $(y_1)^\infty \ksieq{2}
(y_1)^{2^k\omega}v$. Moreover, using  $y_1 \ksieq{2} y_2$
and $v \ksieq{2} v$ together with Lemma~\ref{lem:EFmult}, we obtain
$(y_1)^{2^k\omega}v \ksieq{2} (y_2)^{2^k\omega}v$. Therefore, by
transitivity $(y_1)^\infty \ksieq{2} (y_2)^{2^k\omega}v$. Finally, we
use the fact that $x_1 \ksieq{2} x_2$ and Lemma~\ref{lem:EFmult} to conclude
that $x_1(y_1)^\infty \ksieq{2} x_2(y_2)^{2^k\omega}v$, \emph{i.e.}, that $w_1
\ksieq{2} w_2$.\qed

\subsection{Completeness Proof in Proposition~\ref{prop:sig2main}}
% {\bf Completeness Proof: $\boldsymbol{\Cs_{2,2}[\alpha_\infty] \subseteq
%     \Sat_{\sicd}(\alpha)}$.}

The converse inclusion $\Cs_{2,2}[\alpha_\infty] \subseteq
    \Sat_{\sicd}(\alpha)$ that we show now is more difficult. We devote the rest of this section
to its proof. Before we start the proof, we require two additional results
that we will use.

\medskip\noindent\textbf{Preliminary Results.} The first result that we need is a standard decomposition lemma, which may be
applied to \iwords. We state it in the lemma below.

\begin{lemma} \label{lem:ramsey}
  Let $\gamma: A^+ \to S$ be a morphism into a finite semigroup $S$.
  Then for every \iword $w \in A^\infty$, there exists an idempotent
  $e \in S$ and a decomposition $w = u_0u_1u_2u_3\cdots$ of $w$ into
  infinitely many factors $u_0,u_1,u_2,\dots \in A^+$ satisfying
  $\gamma(u_j) = e$ for all $j \geq 1$ (there is no constraint on
  $u_0$).
\end{lemma}

The proof of Lemma~\ref{lem:ramsey} is standard and is a consequence
of Ramsey Theorem over infinite graphs (see~\cite{womega} for example).

In order to state the second result that we will need, we require some
additional terminology about \ichains. Given $i,k,n \geq 1$,
$x \in \{+,\infty\}$ and $\beta: A^x \to S$ a map into a finite set $S$, we
denote by $\Cs^k_{i,n}[\beta]$ the set of all \chains
$(s_1,\dots,s_n) \in S^n$ for which there exist $w_1,\dots,w_n \in A^x$
satisfying
\begin{itemize}
\item for all $j$, $\beta(w_j) = s_j$.
\item $w_1 \sieq{k}{i} w_2 \sieq{k}{i} \cdots \sieq{k}{i} w_n$.
\end{itemize}
Note that by definition, $\Cs_{i,n}[\beta] = \bigcap_{k\geq 1}
\Cs^k_{i,n}[\beta]$. We will use the following fact, which may be
verified from  Fact~\ref{fct:refine} and finiteness of $S$.

\begin{fact} \label{fct:theell}
  Let $i,n \geq 1$, $x \in \{+,\infty\}$ and let $\beta: A^x \to S$ be a
  map into a finite set~$S$. Then, for all $k \geq 1$,
  \[
    \Cs_{i,n}[\beta] \subseteq \Cs^{k+1}_{i,n}[\beta] \subseteq
    \Cs^k_{i,n}[\beta]
  \]
  In particular, there exists $\ell$ (depending on $i,n$ and $\beta$)
  such $\Cs_{i,n}[\beta] = \Cs^{\ell}_{i,n}[\beta]$.
\end{fact}

Finally, let us mention the following closure properties of $\Cs_i[\beta]$,
which will be used in the proof of the membership problem for \bscd: the set
$\Cs_i[\beta]$ is closed under subwords
and duplication of letters. This is immediate that the definition and the fact
that for all $k$, $\ksieq{i}$ is transitive and reflexive.

\begin{fact} \label{fct:subdup}
  Set $x \in \{+,\infty\}$, $\beta: A^x \to S$ a map into a finite set
  $S$ and let $(s_1,\dots,s_n) \in \Cs_i[\beta]$. Then for all $j \leq n$,
  \[
    (s_1,\dots,s_{j-1}, s_{j+1},\dots,s_n) \in \Cs_i[\beta] \text{ and }
    (s_1,\dots,s_{j-1},s_j,s_j,s_{j+1},\dots,s_n) \in \Cs_i[\beta]
  \]
\end{fact}

\begin{remark}
  A simple consequence of Fact~\ref{fct:subdup} is that, by Higman's
  lemma, $\Cs_i[\beta]$ is a regular language over the alphabet $S$ (and
  is therefore finitely representable). However, this fact is useless in
  the paper: whether an automata for this regular language can be
  computed is open in the cases that we consider.
\end{remark}

\medskip\noindent\textbf{{Proof of Proposition~\ref{prop:sig2main}}.}
We may now prove the inclusion $\Cs_{2,2}[\alpha_\infty] \subseteq
\Sat_{\sicd}(\alpha)$ in Proposition~\ref{prop:sig2main}. We set
$\alpha: (A^+,A^\infty) \to (S_+,S_\infty)$ as an alphabet compatible
morphism into a finite \isemi $(S_+,S_\infty)$. We exhibit a number
$\ell \geq 1$ such that $\Cs_{2,2}^\ell[\alpha_\infty] \subseteq
\Sat_{\sicd}(\alpha)$. By the inclusions in Fact~\ref{fct:theell},
this will prove that $\Cs_{2,2}[\alpha_\infty] \subseteq
\Sat_{\sicd}(\alpha)$.

We begin with the choice of the number $\ell \geq 1$. We know from
Fact~\ref{fct:theell} that there exists a number $\ell_+$ such that
$\Cs_{2,2}[\alpha_+] = \Cs^{\ell_+}_{2,2}[\alpha_+]$. We assume without
loss of generality that $\ell_+ \geq 2$ (by the inclusions in
Fact~\ref{fct:theell} we may choose $\ell_+$ as large as we want).
Furthermore, we set $p = |S_+|+1$. We define $\ell = \ell_+ + p$.

It now remains to prove that $\Cs_{2,2}^\ell[\alpha_\infty] \subseteq
\Sat_{\sicd}(\alpha)$. Set $(q,q') \in
\Cs_{2,2}^\ell[\alpha_\infty]$, we have to prove that $(q,q') \in
\Sat_{\sicd}(\alpha)$. By definition of $\Sat_{\sicd}(\alpha)$, this
means that we have to find $r_1,r_2,s_1,s_2 \in S_+$ and $t_2 \in
S_\infty$ such that
\begin{equation} \label{eq:goal1}
  \begin{array}{l}
    (r_1,r_2)  \in \Cs_{2,2}[\alpha_+] \\
    (s_1,s_2) \in \Cs_{2,2}[\alpha_+] \\
    t_2 \in \alpha(A^\infty) \text{ with } \content{s_1} = \content{t_2}
  \end{array} \quad \text{ and } \quad
  \begin{array}{lll}
    q & = & r_1(s_1)^\infty \\
    q' & = & r_2(s_2)^\omega t_2
  \end{array}
\end{equation}
We proceed as follows. First, we use the definition of
$\Cs_{2,2}^\ell[\alpha_\infty]$ to obtain two \iwords $w$ and $w'$
of images $q$ and $q'$ such that $w \sieq{\ell}{2} w'$. We
then use the hypothesis $w \sieq{\ell}{2} w'$ together with our
decomposition lemma, Lemma~\ref{lem:EFdecomp}, to split $w$
and $w'$ into factors. Finally, we use this decomposition to find
the appropriate $r_1,r_2,s_1,s_2$ and $t$ such that~\eqref{eq:goal1}
holds.

\medskip
\noindent
{\bf Decomposition of $w$ and $w'$.} Using Lemma~\ref{lem:ramsey}
(with $\alpha_+$ as the morphism $\gamma$) we may decompose $w$ as
an infinite product $w = u_0u_1u_2 \cdots$ ($u_0,u_1,u_2,\ldots \in
A^+$) such that $\alpha(u_1) = \alpha(u_2) = \alpha(u_3) = \cdots$ is
an idempotent $e$ of $S_+$. Furthermore, note that since $\alpha$ is
alphabet compatible, $u_1,u_2,\dots$ all share the same alphabet. Let
$B$ be this alphabet.

We now apply Lemma~\ref{lem:EFdecomp} $p$ times to the \iwords $w
\sieq{\ell}{2} w'$. This yields a decomposition $w' = u'_0u'_1
\cdots u'_{p-1}v$ ($u'_0,u'_1,\dots,u'_{p-1} \in A^+$ and $v \in
A^\infty$) which satisfies the following fact (recall that
$\ell = \ell_+ + p$),
\begin{fact} \label{fct:ineq1}
  For all $j \leq p-1$, $u_j \sieq{\ell_+}{2} u'_j$ and $u_{p}u_{p+1}
  \cdots \sieq{\ell_+}{2} v$.
\end{fact}

\medskip
\noindent
{\bf Construction of $r_1,r_2,s_1,s_2$ and $t_2$.} We may now use the
decomposition of $w$ and $w'$ to construct the appropriate $r_1,
r_2,s_1,s_2$ and $t_2$ such that~\eqref{eq:goal1} holds.

Since $p=|S_+|+1$, by the pigeonhole principle, we
obtain $i < j \leq p-1$ such that $\alpha_+(u'_0\cdots u'_{i}) =
\alpha_+(u'_0 \cdots u'_j)=\alpha_+(u'_0\cdots u'_{i})\alpha_+(u'_{i+1}\cdots
u'_{j})$. Hence, $\alpha_+(u'_0\cdots u'_{i})$ is stable by right multiplication
by $\alpha_+(u'_{i+1}\cdots u'_{j})$. Iterating this equality, we get
\[
  \alpha_+(u'_0 \cdots u'_{i}) = \alpha_+(u'_0\cdots u'_{i}) (\alpha_+(u'_{i+1}\cdots u'_{j}))^\omega.
\]

Set $x_1 = u_0 \cdots u_i \in A^+$, $x_2 = u'_0 \cdots u'_i \in A^+$,
$y_1 = u_{i+1}\cdots u_{j} \in A^+$ and $y_2 = u'_{i+1}\cdots u'_{j} \in
A^+$. Moreover, we set $r_1 = \alpha_+(x_1)$, $r_2 = \alpha_+(x_2)$, $s_1
= \alpha_+(y_1)$ and $s_2  = \alpha_+(y_2)$. Note that by the equality
above, we have
\[
  r_2 = r_2(s_2)^\omega.
\]
Finally, we set $z = u'_{i+1} \cdots u'_pv$ and $t_2 = \alpha_\infty(z)$.

It remains to prove that~\eqref{eq:goal1} holds. By definition,
$s_1 = \alpha_+(u_{i+1}\cdots u_{j})$ is the idempotent $e$, therefore
\[
  q = \alpha_\infty(w) = r_1 (s_1)^\infty.
\]
Moreover, we have $w' = x_2z$, therefore,
\[
  q' = r_2t_2 = r_2(s_2)^\omega t_2.
\]
To conclude that~\eqref{eq:goal1} holds, it remains to prove that
$(r_1,r_2),(s_1,s_2) \in \Cs_{2,2}[\alpha_+]$ and
$\content{s_1} = \content{t_2}$. This is what we do now.

We know from Lemma~\ref{lem:EFmult} and Fact~\ref{fct:ineq1} that $x_1
\sieq{\ell_+}{2} x_2$ and $y_1 \sieq{\ell_+}{2} y_2$. This exactly
says that $(r_1,r_2),(s_1,s_2) \in \Cs^{\ell_+}_{2,2}[\alpha_+]$.
Therefore, by choice of $\ell_+$, we have $(r_1,r_2),(s_1,s_2) \in
\Cs_{2,2}[\alpha_+]$. Finally, it is simple to find a \sicd sentence
of rank $2$ that tests the alphabet of an \iword. Since
$\ell_+ \geq 2$ and $u_{i+1}u_{i+2} \cdots \sieq{\ell_+}{2} z$ (see
Lemma~\ref{lem:EFmult} and Fact~\ref{fct:ineq1}), we therefore have
\[
  \content{t_2} = \content{z} =
  \content{u_{i+1}u_{i+2} \cdots} = B = \content{y_1} = \content{s_1}
\]
This terminates the proof of Proposition~\ref{prop:sig2main}.\qed

\section{\texorpdfstring{A Membership Algorithm for \bscd}{A Membership Algorithm for  B-Sigma2}}
\label{sec:bsig2}
In this section, we present our (\ilang) membership algorithm for
\bscd. The algorithm is stated as a decidable characterization of \bscd over~\iwords.

\begin{theorem} \label{thm:caracbc}
  Let $L \subseteq A^\infty$ be a regular \ilang and \hbox{$\alpha:\!
    (A^+,A^\infty) \to (S_+,S_\infty)$} be the alphabet completion of its
  syntactic morphism. The following are~equivalent:
  \begin{enumerate}
  \item $L$ is definable in \bscd.
  \item $\alpha_\infty$ has bounded \sicd-alternation.
  \item\label{item:1} $\alpha_+$ has bounded \sicd-alternation and $\alpha$ satisfies
    the following equation:
    \begin{equation}
      \begin{array}{c}
        (s_2(t_2)^\omega)^\infty = (s_2(t_2)^\omega)^\omega s_1(t_1)^\infty \\
        \text{for all $(s_1,s_2),(t_1,t_2) \in \Cs_{2,2}[\alpha_+]$ with $\content{s_1} = \content{t_1}$}
      \end{array}\label{eq:bcs}
    \end{equation}
  \end{enumerate}
\end{theorem}
Theorems~\ref{thm:fibscd} and~\ref{thm:fisicd} entail that
Item~\ref{item:1} in Theorem~\ref{thm:caracbc} is~decidable (note however that these
two theorems state difficult results of~\cite{pzqalt} whose proofs are
independent from that of Theorem~\ref{thm:caracbc}). Indeed, Theorem~\ref{thm:fibscd}
exactly states that whether $\alpha_+$ has bounded \sicd-alternation is
decidable. Moreover, once one has the set $\Cs_{2,2}[\alpha_+]$ in hand (which is
computable by Theorem~\ref{thm:fisicd}), deciding Equation~\eqref{eq:bcs} may
be achieved by checking all possible combinations. Therefore, we obtain as a
corollary of Theorem~\ref{thm:caracbc} that the \ilang membership problem for
\bscd is decidable.

\begin{corollary} \label{cor:memb2}
  The \ilang membership problem is decidable for \bscd.
\end{corollary}

It now remains to prove Theorem~\ref{thm:caracbc}.
The most difficult
(and interesting) direction is $3) \Rightarrow 2)$. As we did in the previous section, we
first prove the easier $2) \Rightarrow 1)$ and
$1) \Rightarrow 3)$ directions.

\medskip
\noindent
{\bf Proof of $2) \Rightarrow 1)$.} If $\alpha_\infty$ has bounded \sicd-alternation, we know from
Corollary~\ref{cor:bcmemb} that any \ilang recognized by $\alpha$ is
definable in \bscd. But one of these \ilangs is $L$ itself, since $\alpha$ is the
alphabet completion of its syntactic morphism. Hence, $L$ is
definable in $\bscd$, which completes the proof of this implication.

\medskip
\noindent {\bf Proof of $1) \Rightarrow 3)$.} Assume that $L$ is definable in
\bscd. In particular, this means that every language or \ilang recognized by
$\alpha$ is definable in \bscd (we know from Fact~\ref{fct:synta} that it is
true for the syntactic morphism of $L$, so this is true as well for its
alphabet completion $\alpha$, as one can test the alphabet of a word
in~\bscd).

Since every language recognized by $\alpha$ is definable in \bscd, we
know from Corollary~\ref{cor:bcmemb} that $\alpha_+$ has bounded
\sicd-alternation. It remains to prove that Equation~\eqref{eq:bcs} is
satisfied. Set $(s_1,s_2),(t_1,t_2) \in \Cs_{2,2}[\alpha_+]$ with
$\content{s_1} = \content{t_1}$. We have to prove that
$(s_2(t_2)^\omega)^\infty = (s_2(t_2)^\omega)^\omega s_1(t_1)^\infty$.

As every \ilang recognized by $\alpha$ is definable in
\bscd, we may pick a~\bscd sentence $\varphi$ defining
$\alpha^{-1}((s_2(t_2)^\omega)^\infty)$ and let $k$ be its quantifier
rank. By definition of $s_1,s_2,t_1,t_2$, we have words
$x_1,x_2,y_1,y_2 \in A^+$ such that $x_1 \ksieq{2} x_2$, $y_1
\ksieq{2} y_2$, $\alpha(x_1) = s_1$, $\alpha(x_2) = s_2$, $\alpha(y_1)
= t_1$, $\alpha(y_2) = t_2$ and $\content{x_1} = \content{y_1}$.
Set
\[
  w_1  =  (x_2(y_2)^{2^k\omega})^{2^k\omega} x_1(y_1)^\infty\text{\quad and\quad}
  w_2  =  (x_2(y_2)^{2^k\omega})^\infty
\]
Observe that $\alpha(w_1) = (s_2(t_2)^\omega)^\omega s_1(t_1)^\infty$
and $\alpha(w_2) = (s_2(t_2)^\omega)^\infty$. In particular, this
means that $w_2 \models \varphi$. We prove that $w_1 \kbceq{2} w_2$.
By choice of $k$, this will prove that $w_1 \models
\varphi$ and by definition of $\varphi$, that $\alpha(w_1) =
(s_2(t_2)^\omega)^\infty$. This will exactly mean that
$(s_2(t_2)^\omega)^\infty = (s_2(t_2)^\omega)^\omega s_1(t_1)^\infty$
and conclude the proof.

We prove $w_1 \ksieq{2} w_2$ and $w_2 \ksieq{2} w_1$. Note that since
$\content{x_1} = \content{y_1}$, $x_1 \ksieq{2} x_2$ and $y_1
\ksieq{2} y_2$, we have $\content{x_1} = \content{x_2}
= \content{y_1} = \content{y_2}$ (since the alphabet of a
word may be tested in \sicd). Therefore, we have
$\content{x_1(y_1)^\infty} =  \content{x_2(y_2)^{2^k\omega}}$ and we
may use Corollary~\ref{cor:EF} to conclude that $w_2 \ksieq{2} w_1$.
Conversely, we know that
$\content{(x_2(y_2)^{2^k\omega})^\infty} = \content{y_1}$. Therefore,
we may use Corollary~\ref{cor:EF} again to obtain that $(y_1)^\infty
\ksieq{2} (y_1)^{2^k\omega}(x_2(y_2)^{2^k\omega})^\infty$. That $w_1
\ksieq{2} w_2$ is then immediate from this inequality by
Lemma~\ref{lem:EFmult} and transitivity  of \ksieq{2}.

\medskip\noindent {\bf Proof of $3) \Rightarrow 2)$.}
For this last and more difficult implication, we need to introduce a
preliminary result that we will use in the proof. This result is a
generalized version of Proposition~\ref{prop:sig2main}, which gives
an effective description of the set of \dchains of length $n$ for
all $n \geq 2$ (whereas Proposition~\ref{prop:sig2main} is the
specific case $n = 2$).

\subsection{Generalization of Proposition~\ref{prop:sig2main}}

Essentially, Proposition~\ref{prop:sig2main} gives a description of the set
$\Cs_{2,2}[\alpha_\infty]$ of \dchains of length 2 for $\alpha_\infty$ as a
function of the sets $\Cs_{2,2}[\alpha_+]$ and $\alpha(A^\infty)$. Here, we
generalize this to the set $\Cs_{2,n}[\alpha_\infty]$ for an arbitrary fixed
length $n$ of \dchains. In order to understand this generalization, an important
observation is that $\alpha(A^\infty)$ is exactly the set
$\Cs_{2,1}[\alpha_\infty]$ of \dchains of length 1. In other words,
Proposition~\ref{prop:sig2main} describes the \dchains of length $2$ for
$\alpha_\infty$ as a function of the \dchains of length $2$ for $\alpha_+$ and
of the \dchains of length $1$ for $\alpha_\infty$. Our generalization states
that for all $n \geq 2$, the result still holds when replacing length $2$ by
length $n$ and length $1$ by length $n-1$.

Given any \emph{alphabet compatible} morphism $\alpha: (A^+,A^\infty)
\to (S_+,S_\infty)$ into a finite \isemi and any $n \geq 2$, we denote
by $\Sat_{\sicd,n}(\alpha)$  the set of all pairs
\[
  (r_1(s_1)^\infty,r_2(s_2)^\omega t_2,\dots,r_n(s_n)^\omega t_n) \in (S_\infty)^n
\]
with $(r_1,\dots,r_n)  \in \Cs_{2,n}[\alpha_+]$, $(s_1,\dots,s_n) \in
\Cs_{2,n}[\alpha_+]$, $(t_2,\dots,t_n) \in \Cs_{2,n-1}[\alpha_\infty]$
and $\content{s_1} = \content{t_2}$.

\begin{proposition} \label{prop:sig2gen}
  Let $\alpha: (A^+,A^\infty) \to (S_+,S_\infty)$ be an alphabet
  compatible morphism into a finite \isemi $(S_+,S_\infty)$ and
  $n \geq 2$. Then, $\Cs_{2,n}[\alpha_\infty] = \Sat_{\sicd,n}(\alpha)$.
\end{proposition}

The proof of Proposition~\ref{prop:sig2gen} is a straightforward
generalization of that of Proposition~\ref{prop:sig2main} and is left
to the reader.

\subsection{Proof of $3) \Rightarrow 2)$ in Theorem~\ref{thm:caracbc}}

We now have all the material we need to prove the implication $3)
\Rightarrow 2)$ in Theorem~\ref{thm:caracbc}. The proof is based
on a new object that may be associated to any alphabet compatible
\isemi morphism. We first present this object and then explain why it
may be used to prove that $3) \Rightarrow 2)$.

\medskip
\noindent
{\bf Graph Associated to an \iSemi Morphism.} Let $\alpha:
(A^+,A^\infty) \rightarrow (S_+,S_\infty)$ be an alphabet compatible
morphism into a finite \isemi $(S_+,S_\infty)$. We explain how to
associate a graph $G[\alpha]$ to $\alpha$.

We denote by $S^1_+$ the monoid constructed from $S_+$ as follows.
If $S_+$ is already a monoid (\emph{i.e.}, if it has a neutral element
$1_{S_+}$) then $S_+^1 = S_+$. Otherwise $S_+^1 = S_+ \cup
\{1_{S_+}\}$ where $1_{S_+}$ is defined as an artificial neutral
element. We associate to $\alpha$ a graph $G[\alpha]$ whose set
of nodes is $S_+^1 \times S_\infty$ and whose edges are labeled by
subsets of the alphabet (\emph{i.e.}, elements of $2^A$). Given two nodes
$(s_+,s_\infty)$ and $(t_+,t_\infty)$ and $B \in 2^A$, there is an
edge,
\[
  (s_+,s_\infty) \xrightarrow{B} (t_+,t_\infty)
\]
if and only
if there exist $(p_1,p_2),(q_1,q_2) \in  \Cs_{2}[\alpha_+]$ and $q
\in \alpha(A^+)$ such that $\content{p_1} = \content{q_1} =
\content{q} = B$ and,
\[
  s_+p_1  (q_1)^\infty = s_\infty \qquad \text{and} \qquad s_+  p_2
  (q_2)^\omega q = t_+
\]

Observe that the definition of the edge relation does not depend on
$t_\infty$ in the destination node. Given any alphabet $B \subseteq
A$, we call \emph{$B$-cycle} of $G[\alpha]$ a cycle of $G[\alpha]$ in which
all edges are labeled by $B$. Finally, we say that $G[\alpha]$ is
\emph{recursive} if and only if there exist $B \subseteq A$ and a
$B$-cycle such that this cycle contains two nodes $(s_+,s_\infty)$ and
$(t_+,t_\infty)$ with $s_\infty \neq t_\infty$.

\medskip
\noindent
{\bf Proof of the Theorem.} Now that we have defined the graph
$G[\alpha]$, the implication $3) \Rightarrow 2)$ in
Theorem~\ref{thm:caracbc} is an immediate consequence of the two
following lemmas.

\begin{lemma} \label{lem:recursive1}
  Let $\alpha: (A^+,A^\infty) \rightarrow (S_+,S_\infty)$ as an alphabet
  compatible morphism into a finite \isemi $(S_+,S_\infty)$ and assume
  that $\alpha$ satisfies Equation~\eqref{eq:bcs}. Then $G[\alpha]$ is
  not recursive.
\end{lemma}

\begin{lemma} \label{lem:recursive2}
  Let $\alpha: (A^+,A^\infty) \rightarrow (S_+,S_\infty)$ as an alphabet
  compatible morphism into a finite \isemi $(S_+,S_\infty)$ and assume
  that,
  \begin{enumerate}
  \item $\alpha_+$ has bounded \sicd-alternation.
  \item $\alpha_\infty$ has unbounded \sicd-alternation.
  \end{enumerate}
  Then $G[\alpha]$ is recursive.
\end{lemma}

Before we prove these two lemmas, we use them to conclude the proof of
Theorem~\ref{thm:caracbc}. Let $L \subseteq A^\infty$ be a regular
\ilang and let $\alpha: (A^+,A^\infty) \to (S_+,S_\infty)$ be the
alphabet completion of its syntactic morphism. Assume that condition
$3)$ holds in Theorem~\ref{thm:caracbc}, \emph{i.e.}, $\alpha_+$ has bounded
\sicd-alternation and $\alpha$ satisfies~\eqref{eq:bcs}. We have to
prove that $\alpha_\infty$ has bounded \sicd-alternation. We proceed
by contradiction. Assume that $\alpha_\infty$ has unbounded
\sicd-alternation. We now have three hypotheses:
\begin{enumerate}[itemindent=3ex,leftmargin=4ex,label=($\alph*$)]
\item $\alpha$ satisfies~\eqref{eq:bcs}.
\item $\alpha_+$ has bounded \sicd-alternation.
\item $\alpha_\infty$ has unbounded \sicd-alternation.
\end{enumerate}
Therefore, it follows from Lemma~\ref{lem:recursive1} that $G[\alpha]$
is not recursive and from Lemma~\ref{lem:recursive2} that $G[\alpha]$
is recursive, which is a contradiction. We conclude that
$\alpha_\infty$ has bounded \sicd-alternation which terminates the
proof of Theorem~\ref{thm:caracbc}.

It now remains to prove Lemma~\ref{lem:recursive1} and
Lemma~\ref{lem:recursive2}. We devote a subsection to each proof.

\subsection{Proof of Lemma~\ref{lem:recursive1}}

Let $\alpha: (A^+,A^\infty) \rightarrow (S_+,S_\infty)$ be an alphabet
compatible morphism into a finite \isemi $(S_+,S_\infty)$ that
satisfies~\eqref{eq:bcs}. Consider a $B$-cycle
in $G[\alpha]$ and let $(s_+,s_\infty)$ and $(t_+,t_\infty)$ be two
nodes in this $B$-cycle. We have to prove that $s_\infty = t_\infty$.
We first prove that we may actually assume that $(s_+,s_\infty)$ and
$(t_+,t_\infty)$ are the only nodes in the cycle. This follows from
the next fact.

\begin{fact} \label{fct:transitive}
  For all $B \subseteq A$, $\xrightarrow{B}$ is transitive.
\end{fact}

\begin{proof}
  Let $(r_+,r_\infty)$, $(s_+,s_\infty)$ and $(t_+,t_\infty)$ be three
  nodes such that,
  \[
    (r_+,r_\infty) \xrightarrow{B} (s_+,s_\infty) \xrightarrow{B}
    (t_+,t_\infty)
  \]
  By definition of the left edge, we have $(p_1,p_2),(q_1,q_2) \in
  \Cs_2[\alpha_+]$ and $q \in \alpha(A^+)$ such that $\content{p_1} =
  \content{q_1} = \content{q} = B$, $r_+p_1  (q_1)^\infty = r_\infty$
  and $r_+  p_2  (q_2)^\omega q = s_+$.

  Moreover, it follows from the definition of the right edge that we
  have $q'$ such that $\content{q'} = B$ and $s_+q' = t_+$. Set $q'' =
  qq'$, we now have $\content{p_1} = \content{q_1} = \content{q''} = B$,
  $r_+p_1(q_1)^\infty = r_\infty$ and $r_+p_2(q_2)^\omega q'' = t_+$. We
  conclude that $(r_+,r_\infty) \xrightarrow{B} (t_+,t_\infty)$.
\end{proof}

It is immediate from Fact~\ref{fct:transitive} that,
\[
  (s_+,s_\infty) \xrightarrow{B} (t_+,t_\infty) \quad \text{and} \quad
  (t_+,t_\infty) \xrightarrow{B} (s_+,s_\infty)
\]
By definition of $B$-labeled edges we have $(p_1,p_2),(q_1,q_2),
(p'_1,p'_2),(q'_1,q'_2) \in \Cs_2[\alpha_+]$ and $q,q' \in S_+$ such that,
\begin{enumerate}[label=$\alph*)$]
\item\label{item:a} $\content{p_1} = \content{q_1} = \content{q} = \content{p'_1} =
  \content{q'_1} = \content{q'} = B$.
\item\label{item:b} $s_+p_1 (q_1)^\infty = s_\infty$.
\item\label{item:c} $s_+ p_2 (q_2)^\omega q = t_+$.
\item\label{item:d} $t_+p'_1  (q'_1)^\infty = t_\infty$.
\item\label{item:e} $t_+  p'_2 (q'_2)^\omega q' = s_+$.
\end{enumerate}
We now use these equalities to prove that $s_\infty$ and $t_\infty$
are equal. Using \ref{item:c} and \ref{item:e}, we obtain $s_+=t_+  p'_2
(q'_2)^\omega q'=s_+ p_2 (q_2)^\omega qp'_2 (q'_2)^\omega q'$. Therefore,
$s_+$ is stable by right multiplication by $p_2 (q_2)^\omega qp'_2
(q'_2)^\omega q'$, hence
\begin{equation}
  \label{eq:s}
  \begin{array}{lll}
    s_+&=&s_+(p_2 (q_2)^\omega qp'_2(q'_2)^\omega q')\\
       &=&s_+(p_2 (q_2)^\omega qp'_2(q'_2)^\omega q')^{\omega+1}
  \end{array}
\end{equation}
whence by~\ref{item:a}
\[
  \begin{array}{lll}
    s_\infty & = & s_+(p_2 (q_2)^\omega qp'_2 (q'_2)^\omega
                   q')^{\omega+1} p_1 (q_1)^\infty \\
             & = & s_+p_2 (q_2)^\omega\cdot ( qp'_2 (q'_2)^\omega
                   q'p_2 (q_2)^\omega)^{\omega}\cdot   qp'_2 (q'_2)^\omega
                   q'p_1 (q_1)^\infty
  \end{array}
\]

Now, by closure of $\Cs_2[\alpha_+]$ under product and since
$(p_1,p_2)\in\Cs_2[\alpha_+]$, we obtain that
\[
  (qp'_2 (q'_2)^\omega q'p_1,qp'_2 (q'_2)^\omega q'p_2) \in \Cs_2[\alpha_+].
\]
One can also verify that
$\content{qp'_2 (q'_2)^\omega q'p_1} = B = \content{q_1}$. Therefore,
we may apply~\eqref{eq:bcs} to obtain
\[
  \begin{array}{lll}
    s_\infty & = & s_+p_2 (q_2)^\omega( qp'_2 (q'_2)^\omega q'p_2 (q_2)^\omega)^{\infty} \\
             & = & s_+(p_2 (q_2)^\omega qp'_2 (q'_2)^\omega q')^{\infty}
  \end{array}
\]
We now prove that $t_\infty = s_+(p_2 (q_2)^\omega qp'_2 (q'_2)^\omega
q')^{\infty}$ as well. By \eqref{eq:s} and Items~\ref{item:d} and
\ref{item:c}, we get
\[
  \begin{array}{lll}
    t_\infty & = & s_+(p_2 (q_2)^\omega qp'_2 (q'_2)^\omega
                   q')^{\omega+1}p_2 (q_2)^\omega q p'_1 (q'_1)^\infty \\
             & = & s_+p_2 (q_2)^\omega qp'_2 (q'_2)^\omega\cdot(q'p_2
                   (q_2)^\omega qp'_2 (q'_2)^\omega)^{\omega}\cdot q' p_2 (q_2)^\omega q p'_1 (q'_1)^\infty
  \end{array}
\]

Since by hypothesis, we have $p'_1,q'_1\in \Cs_2[\alpha_+]$, we get by closure
under product that
\[
  (q' p_2 (q_2)^\omega q p'_1,q'p_2
  (q_2)^\omega qp'_2) \in \Cs_2[\alpha_+].
\]
One can also verify that
$\content{q' p_2 (q_2)^\omega q p'_1} = B = \content{q'_1}$. Therefore,
we may apply~\eqref{eq:bcs} to obtain,
\[
  \begin{array}{lll}
    t_\infty & = & s_+p_2 (q_2)^\omega qp'_2 (q'_2)^\omega(q'p_2
                   (q_2)^\omega qp'_2 (q'_2)^\omega)^{\infty} \\
             & = & s_+(p_2 (q_2)^\omega
                   qp'_2 (q'_2)^\omega q')^{\infty}
  \end{array}
\]
We conclude that $s_\infty = t_\infty$ which terminates the proof.\qed

\subsection{Proof of Lemma~\ref{lem:recursive2}}

Let $\alpha: (A^+,A^\infty) \rightarrow (S_+,S_\infty)$ be an alphabet
compatible morphism into a finite \isemi $(S_+,S_\infty)$ such that,
\begin{enumerate}[itemindent=3ex,leftmargin=4ex]
\item $\alpha_+$ has bounded alternation.
\item $\alpha_\infty$ has unbounded alternation.
\end{enumerate}
We have to prove that $G[\alpha]$ is recursive. Let $B \subseteq
A$. We say that a node $(s_+,s_\infty)$ of $G[\alpha]$ is a
\emph{$B$-generator} if there exists a \sicd-alternating pair
$(s_1,s_2) \in (S_\infty)^2$ for $\alpha_\infty$ such that $s_\infty =
s_+s_1 \neq s_+s_2$ and $\content{s_1} = B$.

\begin{lemma} \label{lem:alt1}
  $G[\alpha]$ contains at least one $B$-generator for some $B \subseteq A$.
\end{lemma}

\begin{proof}
  This is because $\alpha_\infty$ has unbounded \sicd-alternation. By
  definition this means that there exists at least one \sici-alternating
  pair $(s_1,s_2) \in S_\infty$ for $\alpha_\infty$ such that $s_1 \neq
  s_2$. It is then immediate that $(1_{S_+},s_1)$ is an
  $\content{s_1}$-generator.
\end{proof}

Set $B$ as a minimal alphabet (with respect to inclusion) such
that there exists a $B$-generator. That $G[\alpha]$ is recursive
is a consequence of the next lemma.

\begin{lemma} \label{lem:alt2}
  Let $(s_+,s_\infty)$ be any $B$-generator of $G[\alpha]$. Then, there
  exists a node $(t_+,t_\infty)$ such that

  \begin{enumerate}[itemindent=3ex,leftmargin=4ex]
  \item $s_\infty \neq t_\infty$.
  \item $(t_+,t_\infty)$ is a $B$-generator.
  \item there is a $B$-labeled edge from $(s_+,s_\infty)$ to $(t_+,t_\infty)$.
  \end{enumerate}
\end{lemma}

Since $G[\alpha]$ is a finite graph it is immediate from
Lemma~\ref{lem:alt2} and our choice of $B$ that it must contain a
$B$-cycle in which there are two nodes $(s_+,s_\infty)$ and
$(t_+,t_\infty)$ such that $s_\infty \neq t_\infty$. We conclude that
$G[\alpha]$ is recursive. It now remains to prove
Lemma~\ref{lem:alt2}. In particular, this is where we use
Proposition~\ref{prop:sig2gen}. We devote the remainder of this
section to this proof.

Let $(s_+,s_\infty)$ be a $B$-generator of $G[\alpha]$ and let
$(s_1,s_2) \in (S_\infty)^2$ be the \sicd-alternating pair such that
$s_\infty = s_+s_1 \neq s_+s_2$. We have to construct $(t_+,t_\infty)$
satisfying the conditions of Lemma~\ref{lem:alt2}. We proceed as
follows. First we choose an integer $\ell \geq 1$ and use the fact
that $(s_1,s_2)$ is \sicd-alternating to conclude that $(s_1,s_2)^{\ell}
\in \Cs_{2,2\ell}[\alpha_\infty]$. We then use this result together
with Proposition~\ref{prop:sig2gen} to construct the desired node
$(t_+,t_\infty)$.

Let us begin with the choice of $\ell$. This choice is based on the
two following facts.

\begin{fact} \label{fct:choose1}
  There exists $n_+ \geq 1$ such that for any $t_1,t_2 \in
  S_+$ and any $n \geq n_+$, if $(t_1,t_2)^n \in
  \Cs_{2}[\alpha_+]$, then $(t_1,t_2)$ is \sicd-alternating.
\end{fact}

\begin{fact} \label{fct:choose2}
  There exists $n_\infty \geq 1$ such that for any $t_1,t_2 \in
  S_\infty$ and any $n \geq n_\infty$, if $(t_1,t_2)^n \in
  \Cs_2[\alpha_\infty]$, then $(t_1,t_2)$ is \sicd-alternating.
\end{fact}

The two facts share identical proofs. We show the first one (it suffices to
replace $\alpha_+$ by $\alpha_\infty$ and $\Cs_2[\alpha_+]$ by
$\Cs_2[\alpha_\infty]$ to obtain the second one). If for all $t_1,t_2\in S_+$
and $n\geq1$,
we have $(t_1,t_2)^n\in\Cs_2[\alpha_+]$, then choose $n_+=1$.
Otherwise, since $\Cslev 2[\alpha_+]$ is closed under subwords
(Fact~\ref{fct:subdup}), if $(t_1,t_2)^k\notin \Cs_2[\alpha_+]$, then for all
$j\geqslant k$, we have $(t_1,t_2)^j\notin \Cs_2[\alpha_+]$ as well. Therefore,
one can define $n_+$ as the largest integer $k$ such that there exist
$t_1,t_2 \in S_+$ with $(t_1,t_2)^{k-1} \in \Cs_ 2[\alpha_+]$ but
$(t_1,t_2)^{k} \not\in \Cs_ 2[\alpha_+]$ (with the convention that
$(t_1,t_2)^{0} \in \Cs_ 2[\alpha_+]$).

\medskip
We now set $h = \max(n_+,n_\infty)$ and we choose  $\ell =|S_+|^4 \times
|S_\infty|^2 \times h$. Since $(s_1,s_2)$ is \sicd-alternating, we have in
particular $(s_1,s_2)^{\ell} \in \Cs_{2,2\ell}[\alpha_\infty]$. We now
use this fact together with Proposition~\ref{prop:sig2gen} to
construct $(t_+,t_\infty)$.

Since $(s_1,s_2)^{\ell} \in \Cs_{2,2\ell}[\alpha_\infty]$, we may use
Proposition~\ref{prop:sig2gen} to obtain $(p_1,\dots,p_{2\ell}) \in
\Cs_{2,2\ell}[\alpha_+]$, $(q_1,\dots,q_{2\ell}) \in
\Cs_{2,2\ell}[\alpha_+]$ and $(t_2,\dots,t_{2\ell}) \in
\Cs_{2,2\ell-1}[\alpha_\infty]$ such that:
\begin{itemize}
\item $\content{t_2} = \content{q_1}$
\item $p_1(q_1)^\infty = s_1$.
\item for all $i \geq 1$, $p_{2i}(q_{2i})^\omega t_{2i} = s_2$.
\item for all $i \geq 1$, $p_{2i+1}(q_{2i+1})^\omega t_{2i+1} = s_1$.
\end{itemize}

We set $C = \content{p_1}$ and $D = \content{q_1} =
\content{t_2}$. Observe that since $p_1(q_1)^\infty = s_1$, we have $C
\cup D = \content{s_1} = B$.

Using the pigeonhole principle and our choice of $\ell$, we
obtain a set of $h$ indices $I =\{i_1,\dots,i_h\}$ such that for all
$i,j \in I$,
\[
  \begin{array}{lll}
    p_{2i} & = & p_{2j}\\
    p_{2i+1}&  = & p_{2j+1}
  \end{array}\qquad
  \begin{array}{lll}
    q_{2i} & = & q_{2j}\\
    q_{2i+1} & = & q_{2j+1}
  \end{array}\qquad
  \begin{array}{lll}
    t_{2i} & = & t_{2j}\\
    t_{2i+1} & = & t_{2j+1}
  \end{array}
\]
We define, $p'_2 = p_{2i}$, $p'_{3} = p_{2i+1}$, $q'_2 = q_{2i}$,
$q'_3 = q_{2i+1}$, $t'_2 = t_{2i}$ and $t'_3 = t_{2i+1}$ (for $i \in
I$). Since \dchains are closed under subwords (see
Fact~\ref{fct:subdup}) and $I$ is of size $h$, we know that
\[
  \begin{array}{lll}
    (p_1,p'_2,p'_3,\dots,p'_2,p'_3) & \in & \Cs_{2,2h+1}[\alpha_+] \\
    (q_1,q'_2,q'_3,\dots,q'_2,q'_3) & \in & \Cs_{2,2h+1}[\alpha_+] \\
    (t'_{2},t'_{3},\dots,t'_2,t'_3)  & \in & \Cs_{2,2h}[\alpha_\infty]
  \end{array}
\]
In particular, this means that $(p'_2,p'_3)^h \in
\Cs_{2,2h}[\alpha_+]$, $(q'_2,q'_3)^h \in \Cs_{2,2h}[\alpha_+]$ and
$(t'_2,t'_3)^h \in \Cs_{2,2h}[\alpha_\infty]$. By choice of $h$, it
follows that $(p'_2,p'_3)$, $(q'_2,q'_3)$ and $(t'_2,t'_3)$ are
\sicd-alternating for $\alpha_+$ and $\alpha_\infty$. Furthermore,
since $\alpha_+$ has bounded \sicd-alternation, it follows that
$p'_2 = p'_3$ and $q'_2 = q'_3$. Let us summarize what we
have so far. We have $(p_1,p'_2),(q_1,q'_2) \in \Cs_2[\alpha_+]$ and
$(t'_2,t'_3) \in \Cs_2[\alpha_\infty]$ which is \sicd-alternating for
$\alpha_\infty$ such that
\begin{enumerate}[itemindent=3ex,leftmargin=4ex]
\item $s_1 = p_1(q_1)^\infty$.
\item $s_2 = p'_2(q'_2)^\omega t'_2$.
\item $s_1 = p'_2(q'_2)^\omega t'_3$.
\end{enumerate}

We may now define the node $(t_+,t_\infty)$. Set $t_+ = s_+
p'_2(q'_2)^\omega$ and $t_\infty = s_+p'_2(q'_2)^\omega t'_2$. We
prove  that this node satisfies the conditions of
Lemma~\ref{lem:alt2}. We have $t_\infty = s_+s_2$, which is
different from $s_+s_1 = s_\infty$ by definition of $(s_1,s_2)$.
Hence, we have $s_\infty \neq t_\infty$. Moreover, we have
$(t'_2,t'_3) \in \Cs_2[\alpha_\infty]$ which is \sicd-alternating and,
\[
  \begin{array}{lllllll}
    t_+t'_2 & = & s_+p'_2(q'_2)^\omega t'_2 & = & s_+s_2 & = & t_\infty \\
    t_+t'_3 & = & s_+p'_2(q'_2)^\omega t'_3 & = & s_+s_1 & = & s_\infty
  \end{array}
\]
and we already know that $s_\infty \neq t_\infty$. Therefore,
$(t_+,t_\infty)$ is an $\content{t'_2}$-generator. Furthermore,
is it simple to verify that as an element of the  \dchain,
$(t_2,\dots,t_\ell)$, $t'_2$ has the same alphabet as $t_2$,
\emph{i.e.},  $\content{t'_2} = D \subseteq B$. It follows that
$(t_+,t_\infty)$ is a $D$-generator and by minimality of $B$
that $D = B$.

Finally we prove that there is a $B$-labeled edge from
$(s_+,s_\infty)$ to $(t_+,t_\infty)$. Observe that
\begin{itemize}
\item $s_+ (p_1(q_1)^\omega)(q_1)^\infty = s_\infty$.
\item $s_+  (p'_2(q'_2)^\omega)(q'_2)^\omega (q'_2)^\omega = t_+$.
\end{itemize}
Therefore, since we already know that
$(p_1(q_1)^\omega,p'_2(q'_2)^\omega),(q_1,q'_2) \in \Cs_2[\alpha_+]$,
it suffices to prove that $\content{p_1(q_1)^\omega} = \content{q_1} =
\content{(q'_2)^\omega} = B$ to conclude that there is a $B$-labeled
edge from $(s_+,s_\infty)$ to $(t_+,t_\infty)$. By definition, we have
$\content{p_1(q_1)^\omega} = C \cup D = B$. Similarly, we have
$\content{q_1} = D$ and we have already established that $D = B$.
Finally, since $(q_1,q'_2) \in \Cs_2[\alpha_+]$, we have
$\content{q'_2} = \content{q_1} = B$ and $\content{(q'_2)^\omega} = B$,
which terminates the proof.\qed

\section{\texorpdfstring{A Separation Algorithm for \sict}{A
    Separation Algorithm for Sigma3}}
\label{sec:sig3}
We now present our algorithm for the
\ilang separation problem associated to \sict. As for
\sicd, this algorithm is based on Theorem~\ref{thm:chainsep}: we
present a procedure for computing $\Cs_{3,2}[\alpha_\infty]$ from an
input morphism $\alpha: (A^+,A^\infty) \to (S_+,S_\infty)$.

However, in this case, this computation requires a new
ingredient. This new ingredient is a generalization of \ichains
that we call \emph{mixed \chains}.

\medskip
\noindent
{\bf Mixed \Chains.} Set $x \in \{+,\infty\}$ and $\beta: A^x  \to S$
as a map into some finite set~$S$. We define a set $\Ms[\beta]
\subseteq S^3$. Let $\bar{s} = (s_1,s_2,s_3) \in S^3$ be  a \chain
over $S$. We have $\bar{s} \in \Ms[\beta]$ if and only if for all $k
\in \nat$, there exist $w_1,w_2,w_3 \in A^x$ such that,
\[
  \beta(w_1) = s_1, \ \beta(w_2) = s_2, \ \beta(w_3) = s_3
  \text{\quad and\quad}
  w_1 \ksieq{2} w_2 \ksieq{3} w_3
\]

Mixed \chains were first introduced in~\cite{pseps3} under a different
name: ``$\Sigma_{2,3}$-trees'' (here, 2 and 3 denote \sicd and \sict). In fact ``$\Sigma_{2,3}$-trees'' are a
more general notion: what we call mixed \chains are a particular case
of ``$\Sigma_{2,3}$-trees''.
Essentially $\Sigma_{2,3}$-trees are
trees of depth $3$ whose nodes are labeled by elements of a finite set~$S$ and mixed \chains are the special case when there is only a single
branch in the~tree.

An important remark is that we
will not present any algorithm for computing mixed \chains. On the
other hand, our algorithm for computing $\Cs_{3,2}[\alpha_\infty]$
from an \isemi morphism $\alpha$ is parametrized by the set of mixed
\chains $\Ms[\alpha_+]$. That $\Ms[\alpha_+]$ may be computed from $\alpha_+$
is a very difficult result of~\cite{pseps3}, stated~below.

\begin{theorem}[\cite{pseps3}] \label{thm:fisict2}
  Given as input a morphism $\alpha: A^+ \to S$ into a finite semigroup
  $S$, one can compute the set $\Ms[\alpha]$ of mixed \chains for
  $\alpha$.
\end{theorem}

We may now present our separation algorithm for \sict over \iwords.
Set $\alpha: (A^+,A^\infty) \to (S_+,S_\infty)$ as an alphabet
compatible morphism into a finite \isemi $(S_+,S_\infty)$. We define
$\Sat_{\sict}(\alpha) \subseteq (S_\infty)^2$ as the set of all
pairs
\[
  \big(r_2(s_2(t_2)^\omega)^\infty,\ r_3(s_3(t_3)^\omega)^\omega s_1(t_1)^\infty\big)
\]
with $(r_2,r_3)  \in \Cs_{3,2}[\alpha_+]$, $(s_1,s_2,s_3) \in
\Ms[\alpha_+]$, $(t_1,t_2,t_3) \in \Ms[\alpha_+]$ and $\content{s_1} =
\content{t_1}$. Since we know from
Theorem~\ref{thm:fisict2} that one may compute $\Ms[\alpha_+]$ from
$\alpha$, it is immediate from the definition that one may compute
$\Sat_{\sict}(\alpha)$ from $\alpha$.

\begin{proposition} \label{prop:sig3main}
  Let $\alpha: (A^+,A^\infty) \to (S_+,S_\infty)$ be an alphabet
  compatible morphism into a finite \isemi $(S_+,S_\infty)$. Then,
  $\Cs_{3,2}[\alpha_\infty] = \Sat_{\sict}(\alpha)$.
\end{proposition}
As for \sicd, it is immediate that Proposition~\ref{prop:sig3main}
yields an algorithm for \sict-separation over \iwords. Indeed, it
provides an algorithm computing
$\Cs_{3,2}[\alpha_\infty]$ from any alphabet compatible morphism
$\alpha$, which suffices to decide \sict-separation.

\begin{corollary} \label{cor:sig3sep}
  The \ilang separation problem is decidable for \sict.
\end{corollary}

It remains to prove Proposition~\ref{prop:sig3main}. We proceed as
for \sicd. We first prove the easiest inclusion $\Cs_{3,2}[\alpha_\infty] \supseteq \Sat_{\sict}(\alpha)$.

\medskip
\noindent
{\bf Proof of the inclusion $\boldsymbol{\Cs_{3,2}[\alpha_\infty] \supseteq
    \Sat_{\sict}(\alpha)}$.} Let $(r_2,r_3)  \in
\Cs_{3,2}[\alpha_+]$, $(s_1,s_2,s_3) \in \Ms[\alpha_+]$ and
$(t_1,t_2,t_3) \in \Ms[\alpha_+]$ be \chains such that $\content{s_1} =
\content{t_1}$. We have to prove that
the pair $(r_2(s_2(t_2)^\omega)^\infty,r_3(s_3(t_3)^\omega)^\omega
s_1(t_1)^\infty) $ belongs to $ \Cs_{3,2}[\alpha_\infty]$. Set $k \geq 1$, we
need to find two \iwords $w_2 \ksieq{3} w_3$ such that $\alpha(w_2) =
r_2(s_2(t_2)^\omega)^\infty$ and $\alpha(w_3) = r_3(s_3(t_3)^\omega)^\omega
s_1(t_1)^\infty$.
The definition gives words
$x_2,x_3,y_1,y_2,y_3,z_1,z_2,z_3$ such~that
\begin{itemize}
\item $\alpha(x_j)=r_j$, $\alpha(y_j)=s_j$, $\alpha(z_j)=t_j$
\item $x_2\ksieq{3} x_3$, $y_1\ksieq{2} y_2\ksieq{3} y_3$ and
  $z_1\ksieq{2} z_2 \ksieq{3} z_3$.
\end{itemize}
Moreover, as $\content{s_1}=\content{t_1}$, we have $\content{y_1} =
\content{z_1}$. We define $w_2 = x_2(y_2(z_2)^{2^k\omega})^\infty$ and
$w_3 = x_3(y_3(z_3)^{2^k\omega})^{2^k\omega} y_1z_1^\infty$. It is
immediate from this definition that $\alpha(w_2) =
r_2(s_2(t_2)^\omega)^\infty$ and that $\alpha(w_3) = r_3(s_3(t_3)^\omega)^\omega
s_1(t_1)^\infty$. It remains to prove that $w_2 \ksieq{3} w_3$.

We first prove  $y_1z_1^\infty\ksieq{2}
(y_2(z_2)^{2^k\omega})^\infty$. Since $\content{y_1} =
\content{z_1}$, we may use Corollary~\ref{cor:EF} to obtain
$z_1^\infty \ksieq{2} (z_1)^{2^k\omega} (y_1(z_1)^{2^k\omega})^\infty$.
By Lemma~\ref{lem:EFmult} and transitivity,
\begin{equation}
  \label{eq:1}
  y_1z_1^\infty \ksieq{2} (y_1(z_1)^{2^k\omega})^\infty\ksieq{2} (y_2(z_2)^{2^k\omega})^\infty
\end{equation}
We may now use~\eqref{eq:1} together with Lemma~\ref{lem:EF} to
obtain that $(y_2(z_2)^{2^k\omega})^\infty\ksieq{3} (y_2(z_2)^{2^k\omega})^{2^k\omega}
y_1z_1^\infty$. Using Lemma~\ref{lem:EFmult} and transitivity again, we obtain
that
\[
  x_2(y_2(z_2)^{2^k\omega})^\infty\ksieq{3}x_3(y_3(z_3)^{2^k\omega})^{2^k\omega}
  y_1z_1^\infty
\]
This exactly says that $w_2 \ksieq{3} w_3$ which concludes the proof.\qed

It remains to prove the second
inclusion of Proposition~\ref{prop:sig3main}, that is,  $\Cs_{3,2}[\alpha_\infty] \subseteq \Sat_{\sict}(\alpha)$.
 Note
that the proof relies on Lemma~\ref{lem:ramsey} and
Fact~\ref{fct:theell} that we presented for the proof of the \sicd algorithm. We will also need an additional
result about mixed \chains (it is essentially the `mixed \chains'
version of Fact~\ref{fct:theell}).

\subsection{Preliminary Result}

Given $x \in \{+,\infty\}$ and $\beta: A^x \to S$ a map into a finite
set $S$, for all $k \geq 1$, we denote by $\Ms^k[\beta]$ the set of
all \chains $(s_1,s_2,s_3) \in S^3$ such that there exist
$w_1,w_2,w_3 \in A^x$ satisfying,
\begin{itemize}
\item for all $j$, $\beta(w_j) = s_j$.
\item $w_1 \sieq{k}{2} w_2 \sieq{k}{3} w_3$.
\end{itemize}
Note that by definition, $\Ms[\beta] = \bigcap_{k\geq 1}
\Ms^k[\beta]$. Moreover, the following fact may be verified
from the definition (this uses Fact~\ref{fct:refine} and the fact
that $S$ is finite).

\begin{fact} \label{fct:theell2}
  Let $x \in \{+,\infty\}$ and let $\beta: A^x \to S$ be a
  map into a finite set $S$. Then, for all $k \geq 1$,
  \[
    \Ms[\beta] \subseteq \Ms^{k+1}[\beta] \subseteq \Ms^k[\beta].
  \]
  Moreover, there exists $\ell$ (depending on $\beta$) such that
  $\Ms[\beta] = \Ms^{\ell}[\beta]$.
\end{fact}

\subsection{Proof of Proposition~\ref{prop:sig3main}}

We now prove that $\Cs_{3,2}[\alpha_\infty] \subseteq
\Sat_{\sict}(\alpha)$ in Proposition~\ref{prop:sig3main}. We follow a
proof template which is similar to the one we used to prove
Proposition~\ref{prop:sig2main}. Let
$\alpha: (A^+,A^\infty) \to (S_+,S_\infty)$ be an alphabet compatible
morphism into a finite \isemi $(S_+,S_\infty)$. We exhibit a number
$\ell \geq 1$ such that $\Cs_{3,2}^\ell[\alpha_\infty] \subseteq
\Sat_{\sict}(\alpha)$ (by Fact~\ref{fct:theell}, this suffices since
$\Cs_{3,2}[\alpha_\infty] \subseteq \Cs_{3,2}^\ell[\alpha_\infty]$ for
any $\ell$).

We begin by choosing the number $\ell$. We know from
Fact~\ref{fct:theell} and Fact~\ref{fct:theell2} that there exists
$\ell_+$ such that $\Ms[\alpha_+] = \Ms^{\ell_+}[\alpha_+]$ and
$\Cs_{3,2}[\alpha_+] = \Cs^{\ell_+}_{3,2}[\alpha_+]$. Moreover, we may
assume without loss of generality that $\ell_+ \geq 2$ (we may choose
$\ell_+$ as large as we want). Set $p = |S_+| +1$. We define $\ell' =
\ell_+ + p^2$ and $\ell = \ell' + p$.

We have to prove that $\Cs_{3,2}^\ell[\alpha_\infty] \subseteq
\Sat_{\sict}(\alpha)$. Set $(q,q') \in \Cs_{3,2}^\ell[\alpha_\infty]$,
we prove that $(q,q') \in \Sat_{\sict}(\alpha)$. By definition, of
$\Sat_{\sict}(\alpha)$, we have to find $r_2,r_3,s_1,s_2,s_3,t_1,t_2$
and $t_3$ in $S_+$ such that,
\begin{equation} \label{eq:goal2}
  \begin{array}{l}
    (r_2,r_3)  \in \Cs_{3,2}[\alpha_+] \\
    (s_1,s_2,s_3) \in \Ms[\alpha_+] \\
    (t_1,t_2,t_3) \in \Ms[\alpha_+] \\
    \content{s_1} = \content{t_1}
  \end{array} \quad \text{ and } \quad
  \begin{array}{lll}
    q & = & r_2(s_2(t_2)^\omega)^\infty \\
    q' & = & r_3(s_3(t_3)^\omega)^\omega s_1(t_1)^\infty
  \end{array}
\end{equation}

The existence of $r_2,r_3,s_1,s_2,s_3,t_1,t_2$ and $t_3$ in $S_+$ is a
consequence of the following lemma.

\begin{lemma} \label{lem:sig3witness}
  There exist words $u_2,u_3,x_2,x_3,y_1,y_2,y_3,z_1,z_2$ and $z_3$ in
  $A^+$ such that
  \begin{enumerate}[itemindent=3ex,leftmargin=4ex]
  \item $x_2 \sieq{\ell_+}{3} x_3$, $y_1 \sieq{\ell_+}{2} y_2
    \sieq{\ell_+}{3} y_3$, $z_1 \sieq{\ell_+}{2} z_2 \sieq{\ell_+}{3}
    z_3$ and $u_2 \sieq{\ell_+}{3} u_3$.
  \item $\content{u_3y_1} = \content{z_1}$.
  \item $q = \alpha(x_2(y_2(z_2)^\omega u_2)^\infty)$ and $q' =
    \alpha(x_3(y_3(z_3)^\omega u_3)^\omega y_1(z_1)^\infty)$.
  \end{enumerate}
\end{lemma}

Before proving Lemma~\ref{lem:sig3witness}, let us explain how we use
it to conclude the proof of Proposition~\ref{prop:sig3main}. We set
\[
  \begin{array}{lll}
    r_2 & = & \alpha(x_2(y_2(z_2)^\omega u_2)^{\omega-1} y_2(z_2)^\omega) \\
    r_3 & = & \alpha(x_3(y_3(z_3)^\omega u_3)^{\omega-1} y_3(z_3)^\omega)
  \end{array} \qquad
  \begin{array}{lll}
    s_1 & = & \alpha(u_3y_1) \\
    s_2 & = & \alpha(u_2y_2) \\
    s_3 & = & \alpha(u_3y_3)
  \end{array}
  \qquad
  \begin{array}{lll}
    t_1 & = & \alpha(z_1) \\
    t_2 & = & \alpha(z_2) \\
    t_3 & = & \alpha(z_3)
  \end{array}
\]
We have to prove that these choices satisfy the conditions
in~\eqref{eq:goal2}. The facts that $(r_2,r_3) \in \Cs_{3,2}[\alpha_+]$,
$(s_1,s_2,s_3) \in \Ms[\alpha_+]$ and $(t_1,t_2,t_3) \in
\Ms[\alpha_+]$ is a simple consequence of the first item in
Lemma~\ref{lem:sig3witness} and our choice of $\ell_+$. Let
us detail the case of $(s_1,s_2,s_3)$ (other cases are handled
similarly). By the first item in Lemma~\ref{lem:sig3witness}, we
know that $y_1 \sieq{\ell_+}{2} y_2 \sieq{\ell_+}{3} y_3$ and
$u_2 \sieq{\ell_+}{3} u_3$. Also notice the second inequality means
that $u_3 \sieq{\ell_+}{2} u_2 \sieq{\ell_+}{3} u_3$ (this is comes
from the last item in Fact~\ref{fct:refine}). We may now apply
Lemma~\ref{lem:EFmult} to obtain that $u_3y_1 \sieq{\ell_+}{2}
u_2y_2 \sieq{\ell_+}{3} u_3y_3$. By definition of $s_1,s_2,s_3$, this
exactly says that $(s_1,s_2,s_3) \in \Cs^{\ell_+}_{3,2}[\alpha_+]$ and
we chose $\ell_+$ so that $\Cs^{\ell_+}_{3,2}[\alpha_+] =
\Cs_{3,2}[\alpha_+]$. We conclude that $(s_1,s_2,s_3) \in
\Ms[\alpha_+]$.

Moreover, that $\content{s_1} = \content{t_1}$ is immediate from the
definitions of $s_1$ and $t_1$ since we have $\content{u_3y_1} =
\content{z_1}$ in Lemma~\ref{lem:sig3witness}.

\noindent
Finally, we prove that $q = r_2(s_2(t_2)^\omega)^\infty$ and $q' =
r_3(s_3(t_3)^\omega)^\omega s_1(t_1)^\infty$. We have:
\[
  \begin{array}{lll}
    r_2(s_2(t_2)^\omega)^\infty & = &
                                      \alpha(x_2(y_2(z_2)^\omega u_2)^{\omega-1}
                                      y_2(z_2)^\omega(u_2y_2(z_2)^\omega)^\infty) \\
                                & = &
                                      \alpha(x_2(y_2(z_2)^\omega u_2)^{\omega-1}
                                      (y_2(z_2)^\omega u_2)^\infty)\\
                                & = &
                                      \alpha(x_2(y_2(z_2)^\omega u_2)^{\infty})\\
                                & = & q  \quad \text{by the third item in Lemma~\ref{lem:sig3witness}}
  \end{array}
\]
\noindent and
\[
  \begin{array}{lll}
    r_3(s_3(t_3)^\omega)^\omega s_1(t_1)^\infty & = & \alpha(x_3(y_3(z_3)^\omega u_3)^{\omega-1}
                                                      y_3(z_3)^\omega(u_3y_3(z_3)^\omega)^\omega u_3y_1(z_1)^\infty)
    \\
                                                & = & \alpha(x_3(y_3(z_3)^\omega u_3)^{\omega-1}
                                                      (y_3(z_3)^\omega u_3)^{\omega+1} y_1(z_1)^\infty)\\
                                                & = & \alpha(x_3(y_3(z_3)^\omega u_3)^{2\omega}
                                                      y_1(z_1)^\infty)\\
                                                & = & \alpha(x_3(y_3(z_3)^\omega u_3)^{\omega}
                                                      y_1(z_1)^\infty)\\
                                                & = & q'  \quad \text{by the third item in Lemma~\ref{lem:sig3witness}}\\
  \end{array}
\]

Therefore, our choices of $r_2,r_3,s_1,s_2,s_3,t_1,t_2$ and $t_3$
satisfy the conditions of~\eqref{eq:goal2}, which exactly means that
$(q,q') \in \Sat_{\sict}(\alpha)$ and we are finished with the proof
of Proposition~\ref{prop:sig3main}.

It now remains to prove Lemma~\ref{lem:sig3witness}. We proceed in
two steps. We first prove the existence of a set of words and \iwords
that satisfy weaker properties than the desired words $u_2$, $u_3$, $x_2$, $
x_3$, $y_1$, $y_2$, $y_3$, $z_1$, $z_2$ and $z_3$ in Lemma~\ref{lem:sig3witness}. We
then use this first set of words and \iwords to construct the desired
words $u_2$, $u_3$, $x_2$, $x_3$, $y_1$, $y_2$, $y_3$, $z_1$, $z_2$ and $z_3$. We devote a
subsection to each step.

\subsection{First Step in the Proof of Lemma~\ref{lem:sig3witness}}

As explained, this first step consists in the construction of a
set of words and \iwords that we will then use in the second step to
construct the words $u_2,u_3,x_2,x_3,y_1,y_2,y_3,z_1,z_2$ and $z_3$ of
Lemma~\ref{lem:sig3witness}. We state the construction in the lemma
below. Recall that we have set $p = |S_+| + 1$, $\ell' = p^2 + \ell_+$
and $\ell = p + \ell$, and that $(q,q') \in
\Cs_{3,2}^\ell[\alpha_\infty]$.

\begin{lemma} \label{lem:sig3inter}
  There exist words $w^{}_0,w^{}_1,w'_0,w'_1 \in A^+$ and an \iword
  $w'_\infty \in A^\infty$ such that:
  \begin{enumerate}[itemindent=3ex,leftmargin=4ex]
  \item $w^{}_0 \sieq{\ell'}{3} w'_0$ and $w^{}_1 \sieq{\ell'}{3} w'_1$.
  \item $w'_\infty \sieq{\ell'}{2} (w^{}_1)^\infty$
  \item $q  = \alpha(w^{}_0(w^{}_1)^\infty)$ and $q' =
    \alpha(w'_0(w'_1)^\omega w'_\infty)$.
  \end{enumerate}
\end{lemma}

We devote the subsection to the proof of Lemma~\ref{lem:sig3inter}. By
definition, $(q,q') \in \Cs_{3,2}^\ell[\alpha_\infty]$, which means that there
exist two \iwords $w,w' \in A^\infty$ of images $q,q'$ under $\alpha$ such
that $w \sieq{\ell}{3} w'$. We construct $w_0,w_1, w'_0,w'_1 \in A^+$ and
$w'_\infty \in A^\infty$ by decomposing $w$ and $w'$. Our first move is to
decompose $w$ as an infinite product using Lemma~\ref{lem:ramsey}. However, if
we use $\alpha_+$ only as the morphism for making the decomposition, we will
not have a strong enough link between the factors to prove the desired result.
For this reason, we apply the lemma for a morphism carrying more information
than $\alpha_+$ does. Let us define this morphism.

We know since Lemma~\ref{lem:EFmult} that over $A^+$, the equivalence
$\bceq{\ell'}{3}$ is a congruence for concatenation, hence the quotient
$A^+ / \bceq{\ell'}{3}$ is a semigroup. Moreover, since there are only
finitely many non equivalent \bsct sentences of quantifier rank $\ell'$, this
semigroup is finite. We write
\[
  \begin{array}{rlcl}
    \gamma: & A^+ & \to     & S_+ \times (A^+ / \bceq{\ell'}{3})\\
            & u   & \mapsto & (\alpha(u),[u]_{\bceq{\ell'}{3}})
  \end{array}
\]
\noindent
where $[u]_{\bceq{\ell'}{3}}$ denotes the $\bceq{\ell'}3$-equivalence class of $u$. We now apply
Lemma~\ref{lem:ramsey} to $w \in A^\infty$ with $\gamma$ as the
morphism. We obtain a decomposition $w = u_0u_1u_2 \cdots$
($u_0,u_1,u_2,\dots \in A^+$) such that $\gamma(u_1) = \gamma(u_2) =
\gamma(u_3) = \cdots$ is an idempotent of $S_+ \times (A^+ /
\bceq{\ell'}{3})$. In other words, the decompositions satisfies the two
following properties:
\begin{itemize}
\item there exists an idempotent $e \in S_+$ such that all factors of
  the form $u_iu_{i+1}\cdots u_j$ with $1 \leq i \leq j$ have image
  $e$ under $\alpha$.
\item all factors of the form $u_iu_{i+1}\cdots u_j$ with $1 \leq i
  \leq j$ are $\bceq{\ell'}{3}$-equivalent.
\end{itemize}
We now use this decomposition of $w$ and the hypothesis $w
\sieq{\ell}{3} w'$ to decompose $w'$ as well. Since $w \sieq{\ell}{3}
w'$, we may apply Lemma~\ref{lem:EFdecomp} $p$ times to the $w$ and
$w'$ to obtain a decomposition $w' = u'_0 \cdots u'_{p-1}u'_\infty$ of
$w'$ ($u'_0,u'_1,\dots,u'_{p-1} \in A^+$ and $u'_\infty \in A^\infty$)
which satisfies the following fact (recall that $\ell =\ell' + p$).

\begin{fact} \label{fct:ineq2}
  For all $j \leq p-1$, $u_j \sieq{\ell'}{3} u'_j$ and $u_{p}u_{p+1}
  \cdots \sieq{\ell'}{3} u'_\infty$.
\end{fact}

Since we chose $p=|S_+|+1$, we may now use the pigeonhole principle
to obtain $i < j \leq p-1$ such that $\alpha(u'_0\cdots
u'_{i}) = \alpha(u'_0 \cdots u'_i\cdot u'_{i+1}\cdots u'_{j})$, that is,
$\alpha(u'_0\cdots
u'_{i})$ is stable by right multiplication by $\alpha(u'_{i+1}\cdots u'_{j})$.
This gives us the following fact.
\begin{fact}\label{fct:loop2}
  We have $\alpha(u'_1 \cdots u'_{i}) = \alpha(u'_1\cdots u'_{i}(u'_{i+1}\cdots u'_{j})^\omega)$.
\end{fact}

We may now define the words $w_0,w_1,w'_0,w'_1 \in A^+$ and the \iword
$w'_\infty \in A^\infty$ in the lemma. We set,
\[
  \begin{array}{lll}
    w_0  & = & u_0 \cdots u_i \\
    w'_0 & = & u'_0 \cdots u'_i
  \end{array} \qquad
  \begin{array}{lll}
    w_1  & = & u_{i+1} \cdots u_{j} \\
    w'_1 & = & u'_{i+1} \cdots u'_{j} \\
  \end{array} \qquad
  \begin{array}{lll}
    w'_\infty & = & u'_{i+1} \cdots u'_{p-1}u'_\infty
  \end{array}
\]
It now remains to verify that these choices satisfies the conditions of
the lemma. That $w_0 \sieq{\ell'}{3} w'_0$ and $w_1
\sieq{\ell'}{3} w'_1$ is immediate from Fact~\ref{fct:ineq2} and
Lemma~\ref{lem:EFmult}.

We now prove that $w'_\infty \sieq{\ell'}{2} (w_1)^\infty$. We know from
Fact~\ref{fct:ineq2} and Lemma~\ref{lem:EFmult} that
$u_{i+1}u_{i+2} \dots \sieq{\ell'}{3} w'_\infty$. Moreover, we know by
construction and the choice of the morphism $\gamma$ that
$w_1 = u_{i+1} \cdots u_{j} \bceq{\ell'}{3} u_{h}$ for all $h \geq 1$. In
particular, this means that $w_1 \sieq{\ell'}{3} u_h$ for all $h \geq 1$.
Using Lemma~\ref{lem:EFmult} again, we obtain that
$(w_1)^\infty \sieq{\ell'}{3} u_{i+1}u_{i+2}\cdots$ and by transitivity that
$(w_1)^\infty \sieq{\ell'}{3} w'_\infty$. Finally, we obtain that
$w'_\infty \sieq{\ell'}{2} (w_1)^\infty$ using the last item in
Fact~\ref{fct:refine}.

It remains to prove that  $q  = \alpha(w_0(w_1)^\infty)$ and $q' =
\alpha(w'_0(w'_1)^\omega w'_\infty)$. By definition,
$\alpha(w_0(w_1)^\infty) = \alpha(w_0) e^\infty = \alpha(u_0u_1u_2
\cdots) = \alpha(w) = q$. Moreover, $w'_0w'_\infty = w'$ by
definition. Therefore, $\alpha(w'_0w'_\infty) = \alpha(w') =
q'$. Finally, since $\alpha(w'_0) = \alpha(w'_0(w'_1)^\omega)$
(this is Fact~\ref{fct:loop2}), we obtain that
$\alpha(w'_0(w'_1)^\omega w'_\infty) = q'$ which terminates the proof
of Lemma~\ref{lem:sig3inter}.

This finishes the first step in the proof of
Lemma~\ref{lem:sig3witness}. We present the second and last step in
the next subsection.

\subsection{Second Step in the Proof of Lemma~\ref{lem:sig3witness}}

We finish the proof of Lemma~\ref{lem:sig3witness} by using the words
and \iwords obtained from Lemma~\ref{lem:sig3inter} to construct the
desired words $u_2$, $u_3$, $x_2$, $x_3$, $y_1$, $y_2$, $y_3$, $z_1$, $z_2$ and $z_3$.

Let $w_0,w_1,w'_0,w'_1 \in A^+$ and $w'_\infty \in A^\infty$ that
satisfy the conditions of Lemma~\ref{lem:sig3inter}. We begin by using
the hypothesis $w'_\infty \sieq{\ell'}{2} (w_1)^\infty$ to decompose
$w'_\infty$ and $(w_1)^\infty$. We apply Lemma~\ref{lem:ramsey} to
$w'_\infty$ (with $\alpha_+$ as the morphism) to construct a
decomposition $w'_\infty = v'_0v'_1v'_2 \cdots$ such that
$\alpha(v'_1) = \alpha(v'_2) = \cdots$ is an idempotent $f$ of
$S_+$.

\begin{fact} \label{fct:link3}
  For any $h \geq 0$, $\alpha(v'_0 \cdots v'_h)f^\infty = \alpha(w'_\infty)$.
\end{fact}

An other property that we will use is that since $\alpha$ is alphabet
compatible, for all $h \geq 1$ the $v'_h$ have the same alphabet. We
call $B$ this alphabet.

We now use this decomposition of $w'_\infty$ together with the fact
that $w'_\infty \sieq{\ell'}{2} (w_1)^\infty$ to decompose
$(w_1)^\infty$ as well. We apply Lemma~\ref{lem:EFdecomp} $p^2$ times
to the \iwords $w'_\infty$ and $(w_1)^\infty$. This yields a decomposition
$(w_1)^\infty = v_0 \cdots v_{p^2-1}v_\infty$ of $(w_1)^\infty$
($v_0,v_1,\dots,v_{p^2-1} \in A^+$ and $v_\infty \in A^\infty$) which
satisfies the following fact (recall that $\ell' = \ell_+ + p^2$).

\begin{fact} \label{fct:ineq3}
  For all $h \leq p^2-1$, $v'_h \sieq{\ell_+}{2} v_h$ and
  $v'_{p^2}v'_{p^2+1} \cdots \sieq{\ell_+}{2} v_\infty$.
\end{fact}

Observe that since $v_0 \cdots v_{p^2-1}$ is a finite prefix of
$(w_1)^\infty$, there exists a number $n \geq 1$ such that it is a
prefix of $(w_1)^n$. Therefore, there exists $v_{p^2} \in A^+$ such
that $v_0 \cdots v_{p^2-1}v_{p^2} = (w_1)^n$. Furthermore, we know
from the second item in Lemma~\ref{lem:sig3inter} and
Lemma~\ref{lem:EFmult} that,
\[
  v_0 \cdots v_{p^2-1}v_{p^2} = (w_1)^n \sieq{\ell'}{3} (w'_1)^n
\]
Therefore, we may apply Lemma~\ref{lem:EFdecomp} $p^2$ times to
obtain a decomposition $(w'_1)^n = v''_0 \cdots v''_{p^2-1}v''_{p^2}$
of $(w'_1)^n$ that satisfies the following fact (recall that $\ell' =
\ell_+ + p^2$).
\begin{fact} \label{fct:ineq4}
  For all $h \leq p^2$, $v_h \sieq{\ell_+}{3} v''_h$.
\end{fact}

Finally, since $p = |S_+|+1$, we have $p^2 \geq |S_+|^2+1$. Therefore,
we may apply the pigeonhole principle to obtain $i,j$ such
that $0 \leq i < j \leq p^2 -1$, $\alpha(v_0\cdots v_{i}) =
\alpha(v_0 \cdots v_j)$ and $\alpha(v''_0\cdots v''_{i}) =
\alpha(v''_0 \cdots v''_j)$. In particular, we obtain the following
fact.
\begin{fact}\label{fct:loop3}
  We have
  \[
    \begin{array}{lll}
      \alpha(v_0 \cdots v_{i}) & = & \alpha(v_0\cdots v_{i}(v_{i+1}\cdots v_{j})^\omega)\\
      \alpha(v''_0 \cdots v''_{i}) & = & \alpha(v''_0\cdots v''_{i}(v''_{i+1}\cdots v''_{j})^\omega)
    \end{array}
  \]
\end{fact}
We are now ready to construct the words $u_1,u_2,x_2,x_3,y_1,y_2,y_3,z_1,z_2$
and $z_3$ in $A^+$ from Lemma~\ref{lem:sig3witness}. We set,
\[
  \begin{array}{lll}
    x_2 & = & w_0 \\
    x_3 & = & w'_0
  \end{array} \qquad
  \begin{array}{lll}
    y_1 & = & v'_0\cdots v'_i \\
    y_2 & = & v_0\cdots v_i \\
    y_3 & = & v''_{0} \cdots v''_{i}
  \end{array}
  \qquad
  \begin{array}{lll}
    z_1 & = & v'_{i+1}\cdots v'_{j} \\
    z_2 & = & v_{i+1}\cdots v_{j} \\
    z_3 & = & v''_{i+1}\cdots v''_{j}
  \end{array} \qquad
  \begin{array}{lll}
    u_2 & = & v_{i+1} \cdots v_{p^2} \\
    u_3 & = & v''_{i+1} \cdots v''_{p^2}
  \end{array}
\]
It remains to verify that these words satisfy the conditions of the
lemma. We begin with the inequalities. That $x_2 \sieq{\ell_+}{3}
x_3$ is immediate by choice of $w^{}_0,w'_0$ in
Lemma~\ref{lem:sig3inter}. The inequalities $y_1 \sieq{\ell_+}{2} y_2
\sieq{\ell_+}{3} y_3$, $z_1 \sieq{\ell_+}{2} z_2 \sieq{\ell_+}{3}
z_3$ and $u_2 \sieq{\ell_+}{3} u_3$ are immediate from
Fact~\ref{fct:ineq3}, Fact~\ref{fct:ineq4} and
Lemma~\ref{lem:EFmult}.

Let us now prove that  $\content{u_3y_1} = \content{z_1} = B$
(recall that $B$ is the shared alphabet of all $v'_h$ for $h \geq 1$).
Since $y_1$ is a product of $v'_h$ for $h \geq 1$, it is immediate
that $\content{y_1} = B$. Furthermore, using the same argument, it
is also immediate that $\content{z_1} = B$. Therefore, it suffices
to prove that $\content{u_3} \subseteq B$ to conclude that
$\content{u_3y_1} = B=\content{z_1}$. We know that $u_2 \sieq{\ell_+}{3} u_3$,
therefore $\content{u_3} = \content{u_2}$ ($\ell_+ \geq 2$ and the
alphabet may be tested with a \sict sentence of rank $2$). Moreover,
by definition, $u_2$ is a suffix of $(w_1)^p$. Therefore,
$\content{u_2} \subseteq \content{(w_1)^p} = \content{w_1}$. Finally,
we know from Fact~\ref{fct:ineq3} that $v_{p^2}v_{p^2+1} \cdots
\sieq{\ell_+}{2} v_\infty$. Since $v_\infty$ is a suffix of
$(w_1)^\infty$, we have $\content{v_\infty} = \content{w_1} = B$ and
$\content{u_3} = \content{u_2} \subseteq B$.

We finish with the proof of the last item in
Lemma~\ref{lem:sig3witness}: $q = \alpha(x_2(y_2(z_2)^\omega
u_2)^\infty)$ and $q' = \alpha(x_3(y_3(z_3)^\omega u_2)^\omega
y_1(z_1)^\infty)$. First, by definition, we have $x_2 = w_0$ and
$x_3 = w'_0$. Furthermore, we know from Fact~\ref{fct:loop3} that
\[
  \alpha(y_2(z_2)^\omega u_2) = \alpha(y_2u_2) = \alpha((w_1)^n) \quad \text{and}
  \quad \alpha(y_3(z_3)^\omega u_3) = \alpha(y_3u_3) = \alpha((w'_1)^n)
\]
Finally, we obtain from Fact~\ref{fct:link3} that
$\alpha(y_1(z_1)^\infty) = \alpha(w'_\infty)$. By combining all of
this, we obtain,
\[
  \begin{array}{lllll}
    \alpha(x_2(y_2(z_2)^\omega u_2)^\infty) & = &
                                                  \alpha(w_0((w_1)^n)^\infty) & = & \alpha(w_0(w_1)^\infty)\\
    \alpha(x_3(y_3(z_3)^\omega u_2)^\omega
    y_1(z_1)^\infty) & = & \alpha(w'_0((w'_1)^n)^\omega w'_\infty) &= &\alpha(w'_0(w'_1)^\omega w'_\infty)
  \end{array}
\]
We conclude that $q = \alpha(x_2(y_2(z_2)^\omega
u_2)^\infty)$ and $q' = \alpha(x_3(y_3(z_3)^\omega u_2)^\omega
y_1(z_1)^\infty)$ from the third item in Lemma~\ref{lem:sig3inter}.
This terminates the proof of Lemma~\ref{lem:sig3witness}.\qed

\section{Conclusion}
\label{sec:conc}
We proved that for \ilangs, the separation problem is decidable for
\sicd and \sict and the membership problem is decidable for \bscd.
Note that using a theorem of~\cite{pzsucc2}, these results may be
lifted to the variants of the same logics whose signature has been
enriched with a predicate ``$+1$'' that is interpreted as the
successor relation. This means that over \iwords, separation is
decidable for \siwsd and \siwst and membership is decidable for
\bswsd.

A gap remains between languages and \ilangs: we leave open the case
of \sic{4}-membership for \ilangs while it is known to be decidable
for languages~\cite{pseps3}.  The language algorithm was based on
two results: $1)$ the decidability of \sict-separation~\cite{pseps3}
and $2)$ an effective reduction of \sic{i+1}-membership to
\sici-separation~\cite{pzqalt} (which is generic for all $i \geq 1$).
In the \ilang setting, we are missing the second result and it is not
clear if a similar reduction exists.

\bibliographystyle{abbrv}

\end{document}